\documentclass[11pt,a4paper,fleqn]{article}

\usepackage[T1]{fontenc}
\usepackage{lmodern}
\usepackage[margin=2.35cm]{geometry}
\usepackage{amsmath,amssymb,amsthm}
\usepackage{graphicx,booktabs,caption}
\usepackage[authoryear,longnamesfirst]{natbib}
\usepackage[hidelinks]{hyperref}

\allowdisplaybreaks
\graphicspath{{Figures/}{Graphs/}}

\newtheorem{theorem}{Theorem}
\newtheorem{assumption}{Assumption}
\newtheorem{corollary}[theorem]{Corollary}
\newtheorem{lemma}{Lemma}

\theoremstyle{definition}

\renewenvironment{proof}[1][Proof]{\par\noindent\textbf{#1.}\ }{\hfill$\Box$\par}

\def\plim{\mathop{\rm plim}}
\newcommand{\E}[1]{\mathbb{E}\left[#1\right]}

\DeclareMathOperator*{\diag}{diag}

\title{Detection Boundaries for Panel Slope Homogeneity Tests Under Small-Group Heterogeneity}
\author{Antonio Raiola\thanks{Corresponding author. Via Columbia 2, 00133 Rome, Italy. Email: \texttt{antonio.raiola@uniroma2.it}.}\\
  \small University of Rome Tor Vergata
  \and
  Nazarii Salish\thanks{Calle Madrid 126, 28903 Getafe, Spain. Email: \texttt{nsalish@eco.uc3m.es}.}\\
  \small Universidad Carlos III de Madrid}
\date{}
\hypersetup{
  pdftitle={Detection Boundaries for Panel Slope Homogeneity Tests Under Small-Group Heterogeneity},
  pdfauthor={Antonio Raiola and Nazarii Salish}
}

\begin{document}
\maketitle

\begin{abstract}
Empirical researchers often use slope-homogeneity tests to assess whether slopes can be treated as common across units. A key difficulty is that heterogeneity may be concentrated in a small number of units, so that a failure to reject homogeneity may reflect limited power rather than true homogeneity. This issue is quantified by analyzing the power of standard slope-homogeneity tests under a sequence of doubly local alternatives---alternatives in which only small groups of units depart from the common slope and the magnitude of these deviations shrinks with sample size. Detectability is characterized as a function of panel dimensions, the size of the heterogeneous groups, and the rate at which deviations shrink. Conditions are also identified under which small-group heterogeneity is asymptotically undetectable by these tests yet sufficiently consequential to invalidate inference based on pooled slope estimates. The results thus clarify when homogeneity tests are informative and when a failure to reject may provide misleading reassurance about the validity of pooled inference. A Monte Carlo study supports the theoretical findings.
\end{abstract}


\section{Introduction}
Slope homogeneity is a convenient and widely used restriction in panel data analysis. If slopes are common across units, pooling information over the cross section improves precision and yields estimators with a simple interpretation. If slopes are heterogeneous, common-slope estimands may be difficult to interpret or even misleading (see, e.g., \citealp{BreitungSalish2021}). Empirical researchers therefore use slope-homogeneity tests to assess whether slopes are common and to choose between pooled and heterogeneous specifications. A key difficulty is that heterogeneity may be concentrated in a small number of units, so that non-rejection may reflect limited power against such departures rather than genuine homogeneity . This paper quantifies the local power of standard slope-homogeneity tests against departures confined to small groups. Our results characterize, in terms of the panel dimensions and group size, the conditions under which standard tests can detect such localized departures. They also identify conditions under which small-group heterogeneity remains undetectable by these tests yet sufficiently consequential to invalidate inference based on pooled slope estimates.

The literature on testing slope homogeneity is long-standing. A seminal starting point is \citet{swamy1970efficient}, based on the dispersion of individual-slope estimates around a pooled estimator; refinements such as \citet{pesaran2008testing} remain among the most widely used tools for linear panel models. In parallel, several LM-type procedures have been developed in random-coefficient models, where the null of slope homogeneity is equivalent to a zero-variance restriction on the random coefficients. Building on \citet{breusch1979simple} (see also \citealp[Sec.~6.2.2.4]{hsiao2022analysis}), \citet{juhl2014test} propose LM statistics for zero variance, and \citet{breitung2016lagrange} further extend the methodology to empirically relevant settings. Other LM tests assess whether the regressors have predictive power for the restricted residuals under the null of slope homogeneity: \citet{su2013testing} in models with interactive fixed effects, and \citet{lu2017determining}, whose test for the number of latent groups nests slope homogeneity as a special case.

Existing local-power analyses for these tests consider sequences in which the difference between each heterogeneous slope and a common slope parameter vanishes at a rate determined by the panel dimensions $N$ and $T$. Such sequences characterize the magnitude of the departures that the tests can detect when heterogeneity remains sufficiently widespread, but they do not reveal how detectability changes as the heterogeneous group becomes small relative to the cross section. This distinction is already implicit in the existing literature: \citet{pesaran2008testing} note that dispersion-based tests have nontrivial power when the fraction of slopes departing from the common slope is asymptotically nonnegligible; \citet{su2013testing} highlight the effectiveness of residual-based LM tests when individual slopes can be classified into a finite number of groups. Neither framework, however, characterizes the detection boundary when the fraction of heterogeneous units vanishes. 

We address this question by analyzing the power of the test statistics of \citet{pesaran2008testing}, \citet{su2013testing}, and the LM-type test of \citet{breitung2016lagrange} under a sequence of \emph{doubly local alternatives}. These alternatives are doubly local in that $(i)$ only the slopes of $M$ individuals depart from the dominant common slope, and $(ii)$ the slope differences across groups shrink at a controlled (Pitman-type) rate that may depend on $M$, as well as on $N$ and $T$. Under this design, power depends jointly on the group-specific slope differences and group sizes, making it possible to characterize when small-group heterogeneity remains detectable as $M/N\to0$.

We develop the analysis for a benchmark setting with two groups (one dominant majority and one heterogeneous remainder) and for extensions with multiple groups, with either a fixed or growing number of groups. The analysis covers both large-$N$, large-$T$ panels and large-$N$, fixed-$T$ panels.

Our main results deliver a detection boundary by linking the rate at which slope differences vanish, $\gamma$ say, to the size of the heterogeneous group, $M$. In large-$N$, large-$T$ panels, this boundary is given by
\[
\gamma \asymp \frac{M^{1/2}T^{1/2}}{N^{1/4}},
\]
which nests the classical $N^{1/4}\sqrt{T}$ rate---when $M/N \to c>0$---under the Pitman-drift sequences considered by \citet{pesaran2008testing}, \citet{su2013testing}, and \citet{breitung2016lagrange}. Importantly, the boundary shows that the tests remain informative even if $M/N \to 0$, provided $MT \to \infty$ and the boundary condition above holds, thereby enabling detection of small-group alternatives.

When $T$ is fixed, the boundary reduces to $\gamma\asymp M^{1/2}/N^{1/4}$; that is, the minimal detectable group size is of order $\sqrt{N}\gamma^2$. Below this threshold, heterogeneity is undetectable. Finally, when the number of heterogeneous groups increases with $N$, the tests retain nontrivial power and detectability depends on the total number of units outside the dominant majority rather than on the number of groups.

The detection boundary also clarifies the implications of small-group heterogeneity for pooled inference. 
A failure to reject slope homogeneity often motivates practitioners to pool observations and conduct inference under a common-slope specification. 
Although the pooled estimator remains consistent for the dominant-group slope under the doubly local alternatives, this does not guarantee valid inference. 
The ratio $M/\sqrt{N}$ plays a key role in this distinction.
If $M/\sqrt{N}\to\infty$, there exist alternatives against which the tests have trivial power, yet the heterogeneity-induced bias invalidates inference based on the pooled estimator. 
As discussed in Section~4, the consequences depend on whether the parameter of interest is the dominant-group slope or the average slope across individuals.

The results also speak to the literature on latent group structure in panel models, where such departures arise naturally; see, e.g., \citet{bonhomme2015grouped}, \citet{bester2016grouped}, \citet{su2016identifying}, \citet{lu2017determining}, and \citet{dzemski2024confidence} for identification, estimation, and inference on group membership and group-specific parameters. Our contribution is complementary: rather than estimating the latent group structure, we characterize when standard slope-homogeneity tests can detect departures consistent with it. This perspective also connects to recent work on inference when latent groups are not well separated; see \citet{AkgunOkui2025}.

A related but distinct approach is taken by \citet{kutta2024validating}, who develop tests of approximate slope homogeneity under a user-specified tolerance for heterogeneity. Unlike our analysis, their local-power results concern departures from approximate homogeneity rather than the consequences of undetected heterogeneity for inference based on pooled OLS.

The remainder of the paper is organized as follows. Section~2 introduces the model and notation. Section~3 reviews the test statistics. Section~4 derives the local-power results under doubly local alternatives, and Section~5 reports a Monte Carlo study. Proofs are relegated to the appendix.

\section{Preliminaries}\label{sec:model}
Consider the linear panel data model
\begin{align} \label{eq:regmodel}
  y_{i} = x_{i}\beta_{g(i)} + \epsilon_{i}, \quad i=1,\ldots,N,
\end{align}
where $y_i=(y_{i1},\ldots,y_{iT})'$ is the $T$-dimensional vector of dependent variables for unit $i$; $x_i=(x_{i1},\ldots,x_{iT})'$ is the $T\times K$ design matrix of regressors for unit $i$; and $\epsilon_i=(\epsilon_{i1},\ldots,\epsilon_{iT})'$ is the corresponding vector of idiosyncratic errors. The slope parameter, $\beta_{g(i)}$, is group-specific, with unknown group membership determined by the mapping $g:\mathcal{N} \to \{G_1, \ldots, G_P\}$, where $\mathcal{N} = \{1, \ldots, N\}$. The function $g$ assigns unit $i$ to one of the $P$ groups partitioning the sample; that is, $G_p \cap G_q = \emptyset$ for all $p \neq q$, and $\bigcup_{p=1}^P G_p = \mathcal{N}$.

We are interested in testing the presence of group heterogeneity in the slopes by means of existing slope homogeneity tests. The null hypothesis of identical slopes is
\begin{align*}
H_0: \beta_{g(i)} = \beta, \text{ for all }  i\in \mathcal{N},
\end{align*}
versus the alternative that there exists a group of individuals with different slopes,
\begin{equation*}
  H_1: \exists\ G \subset\mathcal{N}:\;G\neq \emptyset\;\text{and}\;\beta_{g(i)}\neq\beta_{g(j)}\;\text{for}\;g(i)=G\;\text{and}\;g(j)\neq G.
\end{equation*}

As discussed in the introduction, a variety of well-established tests for slope homogeneity are available in the literature. To analyze their asymptotic behavior within a unified framework, we adopt a common set of assumptions sufficiently broad to capture their essential features. Throughout, $C$ denotes a generic constant that may vary across expressions, and $\#G$ denotes the cardinality of a generic set $G\subseteq \mathcal{N}$.
\begin{assumption}[Errors]\label{Ass1}\text{ }
$(a)$ $\epsilon_{it}|x_i\sim IID(0,\sigma^2_i)$, where $\max_i\sigma^2_i<C<\infty$, $\min_i\sigma^2_i>c>0$, with $\epsilon_{it}$ and $\epsilon_{js}$ independently distributed for $i\neq j$ and $t\neq s$; $(b)$ $\E{\epsilon _{it}^{6}|x_i}<C<\infty$. 
\end{assumption}
\begin{assumption}[Regressors]\label{Ass2}\text{ }\\
For all $i\in \mathcal{N}$: $(a)$ $\mathbb{E}|x_{it,k}|^{6}<C<\infty $, for all $t=1,...,T$ and $k=1,...,K$; $(b)$ the probability limit $S_{i}=\plim_{T\to\infty}T^{-1} x_{i}'x_{i}$, is a positive definite, non-stochastic and finite matrix; (c) for any subset $G\subseteq\mathcal{N}$ such that $\#G\to\infty$ as $N\to\infty$, the probability limit $S^{(G)}=\plim_{N,T\to\infty}(\#GT)^{-1}\sum_{i\in G}x_{i}'x_{i}$ is a positive definite, non-stochastic and finite matrix.
\end{assumption}
In particular, taking $G=\mathcal{N}$, $S=S^{(\mathcal{N})}$ is a positive definite non-stochastic matrix.

Assumption~\ref{Ass1}(a), corresponding to Assumption 1(i)--(ii) of \citet{pesaran2008testing}, rules out serial correlation and cross-sectional dependence. Together with model~\eqref{eq:regmodel}, these assumptions provide a simplified framework for presenting our results in a unified manner. However, the tests considered here have been developed or extended to accommodate more general error structures, including serial correlation and cross-sectional dependence induced by common factors (see the original papers as well as \citealt{blomquist2013testing} and \citealt{ando2015simple}). Under these more general settings, the corresponding tests can detect local alternatives at the conventional rate $N^{-1/4}T^{-1/2}$, which is encompassed by our results when the fraction of heterogeneous units is asymptotically nonnegligible relative to the sample size (i.e., $M/N\to m_0>0$). 
Although formal extensions require test-specific arguments and are beyond the scope of this paper, we conjecture that the detection boundary remains informative when appropriately adapted versions of the tests are used. The Monte Carlo experiments in Section~\ref{sec:mc} provide supporting numerical evidence under serial correlation and common-factor dependence.

\section{Test Statistics}\label{sec:tests}
We focus on three representative tests for slope homogeneity, each corresponding to a different approach: (i) the dispersion-based $\Delta$ test of \citet{pesaran2008testing} (hereafter, PY); (ii) the residual-based $\mathcal{J}$ test of \citet{su2013testing} (SC); and (iii) the variance-component $LM$ test of \citet{breitung2016lagrange} (BRS). The first measures the dispersion of individual slope estimates around an estimated common slope; the second assesses the
explanatory power of regressors for residuals obtained under slope homogeneity; and the third tests whether the variance components of random slopes vanish.
Under the null, the first two tests are based on quadratic forms of increasing rank and, after centering and standardization, have standard normal limiting distributions. The third is quadratic in a fixed-dimensional score vector and has a chi-squared limiting distribution.

Despite these differences, their local-power analyses share a common structure: the relevant statistics or scores can be decomposed into a stochastic component governing the null limiting distribution, a component determining local power, and asymptotically negligible remainders under the conditions established in the next section.

We present the tests within the framework of Section~\ref{sec:model}. Let $\hat{\beta}_{LS}=(\sum_{i=1}^N x_i'x_i)^{-1} \sum_{i=1}^N x_i'y_i$ denote the pooled least squares estimator. The corresponding residuals under the common-slope restriction are $\hat{\epsilon}_{it}=y_{it}-x_{it}'\hat{\beta}_{LS}$, collected in $\hat{\epsilon}_i=(\hat{\epsilon}_{i1},\ldots, \hat{\epsilon}_{iT})'$. Let $P_i=x_i(x_i'x_i)^{-1}x_i'$ be the projection matrix onto the column space of $x_i$, with $(t,s)$ entry $P_{i,ts}$.

\subsection{The dispersion-based $\Delta$ test} 
The $\Delta$ test of PY refines the \cite{swamy1970efficient} test, originally designed for testing slope homogeneity in large $T$, fixed $N$, scenarios. 
Both tests compare individual slope estimates, $\hat\beta_i=(x_i'x_i)^{-1}x_i'y_i$, with the weighted least squares (WLS) estimator $ \hat\beta_{WLS}(\bar\sigma_i^2) =\left(\sum_{i=1}^N x_i'x_i/\bar\sigma_i^2\right)^{-1}\sum_{i=1}^N x_i'y_i/\bar\sigma_i^2$, using some estimator $\bar{\sigma}^2_i$ of $\sigma^2_i$. The original Swamy statistic is $\mathcal{S}_{SW} = \mathcal{S}(\hat{\sigma}^2_i)$, where
\[
  \mathcal{S}(\bar{\sigma}^2_i) = \sum_{i=1}^N \left(\hat{\beta}_i-\hat{\beta}_{WLS}(\bar{\sigma}^2_i)\right)'
  \left(\frac{x_i'x_i}{\bar{\sigma}^2_i}\right)
  \left(\hat{\beta}_i-\hat{\beta}_{WLS}(\bar{\sigma}^2_i)\right),
\]
and $\hat{\sigma}^2_i = (T-K)^{-1}(y_{i} - x_{i}\hat{\beta}_i)'(y_{i} - x_{i}\hat{\beta}_i)$ is the unrestricted residual variance estimator, consistent for $\sigma^2_i$ as $T\to\infty$ under both $H_0$ and $H_{1}$. This statistic has a $\chi^2$ limiting null distribution with $K(N-1)$ degrees of freedom when $T\to\infty$ and $N$ is fixed. For the large-$N$, large-$T$ framework, PY replace $\hat{\sigma}^2_i$ with the restricted residual variance estimator
\begin{equation*}
  \widetilde{\sigma}^2_i =\frac{(y_i-x_i\hat\beta_{LS})'
                        (y_i-x_i\hat\beta_{LS})}{T-1}.
\end{equation*}
The resulting statistic, $\mathcal{S}_{PY}=\mathcal{S}(\tilde{\sigma}^2_i)$, after centering and standardization, converges in distribution to a standard normal,
\begin{equation}\label{eq:delta}
  \Delta = \frac{\mathcal{S}_{PY}-NK}{\sqrt{2NK}}\overset{d}{\rightarrow}\mathrm{N}(0,1), \quad\text{ as }\quad N,T\to\infty,
\end{equation}
under $H_0$ and appropriate conditions on the relative growth rates of $N$ and $T$.

\subsection{The residual-based $\mathcal{J}$ test.}
An alternative approach is the LM test of SC,\footnote{With additive fixed effects (or without fixed effects, as here), the LM test of SC is numerically equivalent to the test of \cite{lu2017determining} when the null hypothesis specifies a single group.} based on the idea that, under $H_0$ and with predetermined regressors, $x_{it}$ should not have explanatory power for the residuals in \eqref{eq:regmodel}. This motivates a test of $H_0^*: \phi_i=0\;\text{for all } i=1,\ldots,N$ in the auxiliary regression
\[
  \hat{\epsilon}_{it} = \phi_i'x_{it} + \eta_{it}, \qquad \text{with} \quad i=1,\ldots \,N, t=1,\ldots,T,
\]
where $\eta_{it}$ is a mean-zero error term.

The SC test statistic is given by the following centered and standardized $LM$ test statistic, which has a limiting normal distribution under $H_0^*$:
\begin{equation}\label{eq:SC}
  \mathcal{J}= \frac{N^{-1/2}LM_{SC}-\hat{B}}{\sqrt{\hat{\mathcal V}}}\overset{d}{\rightarrow}\mathrm{N}(0,1), \quad N,T\to\infty,
\end{equation}
where
\begin{equation*}
  LM_{SC} = \sum_{i=1}^N\hat{\epsilon}_i'x_i(x_i'x_i)^{-1}x_i'\hat{\epsilon}_i.
\end{equation*}
Here, $\hat{B} = N^{-1/2}\sum_{i=1}^N\sum_{t=1}^T\hat{\epsilon}_{it}^2P_{i,tt}$ provides the null centering and $\hat{\mathcal V} = 4N^{-1}\sum_{i=1}^N\sum_{t=2}^T\left(\hat{\epsilon}_{it}\sum_{s=1}^{t-1}P_{i,ts}\hat{\epsilon}_{is}\right)^2$ estimates the asymptotic variance of the centered statistic under $H_0^*$. 

\subsection{The $LM$ variance test.}
A third approach reformulates the null hypothesis within the random coefficients model. Specifically, assuming $\beta_i=\beta+v_i$ with $\mathbb{E}[v_iv_i'|x_i]=\text{diag}(\sigma^2_{v,1},\dots,\sigma^2_{v,K})$, slope homogeneity is equivalent to $H_0^{**}: \sigma^2_{v,k}=0$, for all $k=1,\dots,K$, which reduces the $K(N-1)$ linear restrictions tested by $\Delta$ and $\mathcal{J}$ to $K$.

Building on this formulation, BRS develop an LM test valid for panels with large $N$ and either fixed $T$ or large $T$. The statistic is based on the score vector
$\hat{s}=(\hat{s}_1,\ldots,\hat{s}_K)'$, where
\begin{equation*}
  \hat{s}_k = \sum_{i=1}^N\sum_{t=2}^T\sum_{s=1}^{t-1}\hat{\epsilon}_{it}\hat{\epsilon}_{is}x_{it,k}x_{is,k},\;\;k=1,...,K.
\end{equation*}
The test statistic is then given by
\begin{equation}\label{eq:LMtest}
  LM = \hat{s}'\hat{V}^{-1}\hat{s},
\end{equation}
where $\hat{V}$ is the variance estimator constructed under the common-slope restriction, with $(k,l)$-th element
\begin{equation*}
  \hat{V}_{kl} = \sum_{i=1}^N \sum_{t=2}^T\sum_{s=1}^{t-1}\sum_{p=1}^{t-1}
  \hat{\epsilon}_{it}^2x_{it,k}x_{it,l}\, \hat{\epsilon}_{is}\hat{\epsilon}_{ip}\, x_{is,k}x_{ip,l} \quad  \quad l,k=1,...,K.
\end{equation*}
Under $H_0^{**}$ and $N\to\infty$, with $T$ either fixed or $T\to\infty$, the test has a $\chi^2_K$ limiting distribution. 

The three tests, $\Delta$, $\mathcal{J}$, and $LM$, have well-established local power properties. However, none has been explicitly studied for local power against small-group heterogeneity, which is the focus of the next section.

\section{Asymptotic Local Power Analysis}\label{sec:Asymptotics}
We study the local power of the $\Delta$, $\mathcal J$, and $LM$ tests, focusing on (i) how large slope differences must be and (ii) how large the groups must be for these differences to be detected. To this end, we introduce a sequence of \textit{doubly local alternatives} that extends the classical Pitman framework and integrates these two aspects. We first present the results for the case where the sample is partitioned into two groups, $\{G_1,G_2\}$, with $N,T\to\infty$. Group $G_1$ is treated as the ``dominant'' group and $G_2$ as the ``alternative'' one, containing the remaining individuals. We then extend these results to panels with $N \to \infty$ and fixed $T$, as well as to settings with more than two groups.

Formally, define
\begin{eqnarray*}
  H_{1,n}: \beta_{g(i)} =
  \begin{cases}
    \beta,                 & \text{if } g(i)= G_1, \\
    \beta + \lambda_2/\gamma, & \text{if } g(i)= G_2,
  \end{cases}
\end{eqnarray*}
where $\lambda_2\neq0$ is a $K\times 1$ vector, and $\{G_1,G_2\}$ form a partition of $\mathcal{N}$. The number of individuals in each group is
\begin{equation*}
  \# G_2 = M, \qquad \# G_1 = N-M,
\end{equation*}
where $M=M(N)$ may either remain fixed or diverge with $N$, and $\gamma=\gamma(N,T,M)\to\infty$ along the asymptotic sequence. We assume that $\lim_{N,T\to\infty} MT/(\sqrt{N}\gamma^2) <\infty$, which rules out consistently detectable alternatives. 

Allowing the scaling factor $\gamma$ to depend on $M$ as well as on $N$ and $T$ reflects that the local alternatives control for both the slope difference and the size of the heterogeneous group. This extension is essential: the relative magnitude of $M$ with respect to $N$ and $T$ determines when slope homogeneity tests can detect group-based heterogeneity. Theorem \ref{theorem:Tlarge} formalizes this statement. For each test, assume that the probability limits stated in the theorem exist and are finite.

\setcounter{theorem}{0}
\begin{theorem} \label{theorem:Tlarge}
Suppose Assumptions~\ref{Ass1} and \ref{Ass2} hold. Under $H_{1,n}$, as $N,T\to\infty$ and $M/N\to m_0\in[0,1)$:
\begin{description}

        \item[(PY test)] Suppose additionally that $\sqrt N/T\to0$,  $\E{|\epsilon_{it}|^9\mid x_i}\leq C$
        a.s., and the PY regularity conditions in \eqref{eq:PYregularity} of Appendix~\ref{sec:proofs} hold. Then
              \begin{equation*}
                \Delta\overset{d}{\rightarrow }\begin{cases}
                  \mathrm{N}(0,1),     & \mbox{if }  \frac{\sqrt{MT}}{\gamma \, N^{1/4}} \to 0,      \\
                  \mathrm{N}(\delta_{\Delta},1), & \mbox{if }  \frac{\sqrt{MT}}{\gamma \, N^{1/4}} \to c>0,
                \end{cases}
              \end{equation*}
              where $${\delta}_{\Delta} = \frac{c^2}{\sqrt{2K}}\lambda_2'\left(Q^{(G_2)}-m_0Q^{(G_2)}Q^{-1}Q^{(G_2)}\right)\lambda_2,$$
              with $Q^{(G)}=\plim\limits_{N,T\to\infty}(\#GT)^{-1}\sum_{i\in G}x_{i}'x_{i}/\sigma^2_i$, and $Q=Q^{(\mathcal{N})}$.
              
        \item[(SC test)]
             \begin{equation*}
                \mathcal{J}\overset{d}{\rightarrow }\begin{cases}
                  \mathrm{N}(0,1),  & \mbox{if }   \frac{\sqrt{MT}}{\gamma \, N^{1/4}} \to 0,      \\
                  \mathrm{N}(\delta_{\mathcal{J}},1), & \mbox{if }  \frac{\sqrt{MT}}{\gamma \, N^{1/4}} \to c>0,
                \end{cases}
              \end{equation*}
              where
              \[
              {\delta}_{\mathcal{J}} = \frac{c^2}{\sqrt{V_0}}\lambda_2'\Bigl[W^{(G_2)}
              +m_0\bigl(S^{(G_2)}S^{-1}WS^{-1}S^{(G_2)}
              -2W^{(G_2)}S^{-1}S^{(G_2)}\bigr)\Bigr]\lambda_2,
              \]
              with $W^{(G)}=\plim\limits_{N,T\to\infty}(\#GT)^{-1}\sum_{i\in G}\sum_{t=1}^{T}{x_{it}x_{it}'}(1-P_{i,tt})$, and $W=W^{(\mathcal{N})}$. Here \\
              $\mathcal V_0=\operatorname*{plim}_{N,T\to\infty} 4N^{-1}\sum_{i=1}^N\sigma_i^4 \sum_{t=2}^{T}\sum_{s=1}^{t-1}P_{i,ts}^2>0$.
        \item[(BRS test)] \begin{equation*}
                LM\overset{d}{\rightarrow }\begin{cases}
                  \chi _{K}^{2},                & \mbox{if } \frac{\sqrt{MT}}{\gamma \, N^{1/4}} \to 0,      \\
                  \chi _{K}^{2}\left(\delta_{LM}'\Psi^{-1}\delta_{LM} \right), & \mbox{if }  \frac{\sqrt{MT}}{\gamma \, N^{1/4}} \to c>0,
                \end{cases}
              \end{equation*}
              where $\delta_{LM}$ is a $K\times1$ vector, with $k$-th element
              \[
                \delta_{LM,k} = c^2\left(\lambda_2^{\prime}\Omega_k^{(G_2)} \lambda_2 + m_0\lambda_2'\left(S^{(G_2)}S^{-1}\Omega_kS^{-1}S^{(G_2)} - 2S^{(G_2)}S^{-1}\Omega_{k}^{(G_2)}\right)\lambda_2\right)
              \]
              with $\Omega_{k}^{(G)} = \plim_{N,T \rightarrow \infty} \; (2\#GT^2)^{-1} \sum_{i\in G}\sum_{t=2}^{T}\sum_{s=1}^{t-1}x_{it,k}x_{is,k}\left(x_{it}x_{is}'+x_{is}x_{it}'\right)$,
              and $\Omega_{k}=\Omega_{k}^{(\mathcal{N})}$. Here, ${\Psi}$ is a $K \times K$ matrix with ${\Psi} _{k,l} = \underset{N,T \rightarrow \infty}{\text{plim}} \; (NT^2)^{-1} \sum_{i=1}^{N} \sum_{t=2}^{T}\sum_{s=1}^{t-1}\sigma^4_i x_{it,k}x_{it,l} x_{is,k}x_{is,l}$.
      \end{description}
\end{theorem}
Theorem \ref{theorem:Tlarge} highlights a trade-off between the rate at which slope differences vanish and the size of the alternative group. The critical rate,
\begin{equation}\label{eq:gammarate}
  \gamma \asymp \frac{M^{1/2}T^{1/2}}{N^{1/4}},
\end{equation}
marks the boundary between trivial and nontrivial power. If $\gamma$ grows faster than this rate, all three tests have trivial power. If it grows at or below this rate, they have power to detect the grouped alternative.

Two cases deserve attention. If the fraction of alternative units remains bounded away from zero ($m_0 > 0$, so $M \asymp N$), the rate \eqref{eq:gammarate} reduces to the classical order $O(\sqrt{T}N^{1/4})$, consistent with the local power results of PY, SC, and BRS. 
When the alternative group is small ($m_0=0$), the rate is slower than $O(\sqrt{T}N^{1/4})$, yet the tests retain nontrivial local power even as the fraction of units with distinct slopes vanishes, provided $MT/\sqrt N\to\infty$ and the slope differences shrink at the boundary rate.
This contrasts with the prevailing assumption in the grouped-patterns literature, which requires heterogeneous slopes to affect a non-negligible share of the sample \citep[e.g.,][]{bonhomme2015grouped,su2016identifying,lu2017determining}.

Another important implication of Theorem~\ref{theorem:Tlarge} concerns the consequences of a failure to reject slope homogeneity. In applications, a non-rejection often motivates pooling and reporting confidence intervals based on a common-slope specification. It is therefore useful to ask whether departures that are not detected by the tests can nevertheless affect subsequent inference. Under $H_{1,n}$, the pooled estimator $\hat{\beta}_{LS}$ remains consistent for the dominant-group slope $\beta$, as shown in Lemma~\ref{lem:Lemma1} in the Appendix. This, however, does not imply that the effect of heterogeneity is negligible relative to the standard error of $\hat{\beta}_{LS}$, as required for inference based on its limiting null distribution.

To see this, strengthen Assumption 2 by requiring $M^{-1}\sum_{i\in G_2}\hat S_i - S^{(G_2)} = O_p(1/\sqrt{MT})$. The decomposition in the proof of Lemma~\ref{lem:Lemma1} then yields
\[
\sqrt{NT}(\hat{\beta}_{LS}-\beta) = Z_{NT} +
\frac{\sqrt{MT}}{\gamma N^{1/4}}
\left(\frac{M}{\sqrt{N}}\right)^{1/2}
S^{-1}S^{(G_2)}\lambda_2 + o_p(1)
,
\]
where $Z_{NT} = (NT)^{-1/2} S^{-1}\sum_{i=1}^{N}\sum_{t=1}^{T}x_{it}\epsilon_{it} +o_p(1)$ is the centered stochastic component that drives the asymptotic distribution of the estimator under $H_0$. The second term scales with the ratio of the detection boundary to $\gamma$, multiplied by $(M/\sqrt{N})^{1/2}$. Three cases arise. If $M/\sqrt{N}\to 0$, alternatives at or below the detection boundary are asymptotically innocuous for inference on the dominant-group slope.
If $M/\sqrt{N}\to r\in(0,\infty)$, alternatives below the detection boundary \eqref{eq:gammarate} remain asymptotically harmless, whereas those at the boundary induce a first-order bias but are detected with nontrivial local power. If $M/\sqrt{N}\to\infty$, however, there exist alternatives below the detection boundary against which the tests have asymptotically trivial power, yet the heterogeneity-induced bias remains first-order relevant or even dominates the standard error. Thus, failure to reject slope homogeneity need not guarantee the validity of inference based on pooled estimates.

The implications differ when the parameter of interest is the average slope across individuals, $\bar\beta = N^{-1}\sum_{i=1}^{N}\beta_{g(i)}$. In this case, under $H_{1,n}$,
\[
\sqrt{NT}(\hat{\beta}_{LS}-\bar\beta) = Z_{NT} + \frac{M\sqrt T}{\gamma\sqrt N} \left(S^{-1}S^{(G_2)}-I\right)\lambda_2 + o_p(1).
\]
Hence, the leading heterogeneity-induced bias is asymptotically negligible if
\[
\frac{M\sqrt T}{\gamma\sqrt N}\left\|\left(S^{-1}S^{(G_2)}-I\right)\lambda_2\right\|=o(1).
\]
Inference on the average slope may therefore remain valid even when inference on the dominant-group slope is invalid. In particular, if the regressor second moments in the alternative group coincide with those in the full sample, $S^{(G_2)}=S$, the bias for the average slope is asymptotically harmless. Otherwise, heterogeneity-induced bias has the same rate as that for the dominant-group slope.

\subsection{On the noncentrality parameters}
Theorem \ref{theorem:Tlarge} shows that all three tests share the same asymptotic detection boundary \eqref{eq:gammarate}. Their local power need not coincide, however, since the noncentrality parameters differ. A general comparison is intractable, as the parameters depend on both model and design in complex ways. We therefore consider the benchmark case of a single regressor, $K=1$, to build intuition.

Let $x_{it}$ be scalar independent across individuals and i.i.d.\ over time for each individual, with $\E{x_{it}}=0$, $\E{x_{it}^2}=\sigma_{i,x}^2$, and uniformly bounded fourth moments. The regressor may thus exhibit unit-specific variance, constant over time. For any subset $G\subseteq \mathcal N$, define
\[
\mu_{2,x}^{(G)}=\lim_{N\to\infty}\frac{1}{\#G}\sum_{i\in G}\sigma_{i,x}^2,\quad
\mu_{4,x}^{(G)}=\lim_{N\to\infty}\frac{1}{\#G}\sum_{i\in G}\bigl(\sigma_{i,x}^2\bigr)^2, \quad q_x^{(G)}=\lim_{N\to\infty}\frac{1}{\#G}\sum_{i\in G}\frac{\sigma_{i,x}^2}{\sigma_i^2},
\]
with $\mu_{2,x}=\mu_{2,x}^{(\mathcal{N})}$, $\mu_{4,x}=\mu_{4,x}^{(\mathcal{N})}$, and $q_x=q_x^{(\mathcal{N})}$.
Now suppose $m_0=0$, so $M/N\to 0$ and the heterogeneous subset is asymptotically negligible. The second term in each noncentrality parameter vanishes, and Theorem \ref{theorem:Tlarge} reduces to
\[
\delta_{\Delta}=\frac{c^2\lambda_2^2}{\sqrt{2}}\,q_x^{(G_2)},
\qquad
\delta_{\mathcal J}=\frac{c^2\lambda_2^2}{\sqrt{\mathcal V_0}}\,\mu_{2,x}^{(G_2)},\quad\text{and}\quad\xi_{LM} = \delta_{LM}'\Psi^{-1}\delta_{LM} =\frac{c^4\lambda_2^4}{4\Psi}\,\left(\mu_{4,x}^{(G_2)}\right)^2,
\]
where $ \mathcal V_0 = 2\lim_{N\to\infty}N^{-1}\sum_{i=1}^N \sigma_i^4$ and $\Psi = \lim_{N\to\infty}(2N)^{-1}\sum_{i=1}^N \sigma_i^4(\sigma_{i,x}^2)^2$. Even in the scalar case, PY and SC depend on different weighted second moments, while BRS depends on fourth-order cross-sectional moments of the regressor.

To compare PY and SC directly, note that
\[
\frac{\delta_{\Delta}}{\delta_{\mathcal{J}}} = \lim_{N\to\infty} \mathrm{RMS}(\sigma^2_i) \sum_{i\in G_2}\frac{\omega_{x,i}}{\sigma^2_i}\;
\]
where $\omega_{x,i}=\sigma^2_{x,i}/\sum_{i\in G_2}\sigma^2_{x,i}$, and $\mathrm{RMS}(\sigma^2_i)=\sqrt{N^{-1}\sum_{i=1}^{N}\sigma^4_i}$. It follows that,
\[
\lim_{N\to\infty}\frac{\mathrm{RMS}(\sigma^2_i)}{\max_{i\in G_2}\sigma^2_i} \leq \frac{\delta_{\Delta}}{\delta_{\mathcal{J}}} \leq \lim_{N\to\infty} \frac{\mathrm{RMS}(\sigma^2_i)}{\min_{i\in G_2}\sigma^2_i}.
\]
When $G_2$ concentrates on units with small $\sigma^2_i$, the $\Delta$ test dominates; otherwise the $\mathcal{J}$ test has higher local power. Under homoskedasticity ($\sigma_i^2=\sigma^2$ for all $i$), $\delta_{\Delta} = \delta_{\mathcal J}$ and the two tests are locally equivalent.

Comparing PY and SC with BRS is more complex. Under homoskedasticity and constant regressor variance ($\sigma_{i,x}^2=\sigma_x^2$ for all $i$), all three procedures are locally equivalent: $\xi_{LM}/\delta_{\Delta}^2 = \xi_{LM}/\delta_{\mathcal J}^2 =1$. If errors are homoskedastic and $\{\sigma_{x,i}^2\}$ unrestricted with $\mu_{2,x}^{(G_2)}=\mu_{2,x}$, and $\mu_{4,x}^{(G_2)}=\mu_{4,x}$, then
\[
\frac{\xi_{LM}}{\delta_{\Delta}^2} = \frac{\mu_{4,x}}{\mu_{2,x}^2} \geq 1.
\]
Hence, BRS has a larger local noncentrality parameter than PY and SC, with the ratio $\xi_{LM}/\delta_\Delta^2$ governed by the cross-sectional dispersion in regressor variance.

A complementary comparison is obtained when regressor variance is constant across units but errors are heteroskedastic. Then, $\delta_{\mathcal J}^2 = \xi_{LM}$, so SC and BRS tests are locally equivalent. If $G_2$ is also asymptotically representative with respect to error variances---i.e., $\lim_{N\to\infty}M^{-1}\sum_{i\in G_2}\sigma_i^{-2} = \lim_{N\to\infty}N^{-1}\sum_{i=1}^N\sigma_i^{-2}$---then
\[
\frac{\xi_{LM}}{\delta_{\Delta}^2}
=
\left(\frac{\delta_{\mathcal J}}{\delta_{\Delta}}\right)^2
=
\lim_{N\to\infty} \frac{1}
{\left(N^{-1}\sum_{i=1}^N \sigma_i^4\right)
\left(N^{-1}\sum_{i=1}^N \sigma_i^{-2}\right)^2}
\leq 1,
\]
with equality under homoskedasticity. When regressor variance is constant, error heteroskedasticity thus favors PY over SC and BRS

When $m_0>0$ (i.e., $M \asymp N$), the comparison is more complex as the second term in each expression of noncentrality parameters is no longer asymptotically negligible.
However, the analysis of the same scalar benchmark points in the same direction. 
In particular, under homoskedasticity, $\delta_{\Delta} = \delta_{\mathcal J}$; if, in addition, the regressor variance is constant, all three tests are locally equivalent; if, instead, regressor variance is unrestricted, $\mu_{2,x}^{(G_2)}=\mu_{2,x}$ and $\mu_{4,x}^{(G_2)}=\mu_{4,x}$, then
\[
\frac{\xi_{LM}}{\delta_{\Delta}^2} = \frac{\mu_{4,x}}{\mu_{2,x}^2} \geq 1.
\]
Outside these benchmarks, the comparison depends on the joint distribution of $\{\sigma^2_i\}$ and $\{\sigma^2_{x,i}\}$ across the full sample and the remainder group $G_2$.

\subsection{Panels with $N\to\infty$ and $T$-fixed}
When $T$ is fixed and $N$ grows, the LM-type test of BRS is the only one available. The BRS statistic requires $T\geq 2$, a condition maintained throughout this subsection. The restrictions on the regressors can be relaxed, as specified in Assumption \ref{Ass3} in the Appendix, which also provides the corresponding version of Theorem \ref{theorem:Tlarge}. In this setting, the boundary between trivial and nontrivial power is governed by the rate
\begin{equation}\label{eq:gammarate2}
\gamma \asymp \frac{\sqrt{M}}{N^{1/4}}, 
\end{equation}
the analogue of \eqref{eq:gammarate} for fixed $T$.

When the alternative group is proportional to the total cross-section ($m_0 > 0$, or $M \!\asymp\! N$), the rate in \eqref{eq:gammarate2} reduces to $\gamma \!\asymp\! N^{1/4}$, matching the local-power behavior established in BRS. More interestingly, this rate marks the boundary at which the test begins to exhibit power. If $M$ diverges at least as fast as $N^{1/2}\gamma^2$, the $LM$ test retains nontrivial local power even when $M/N \!\to\! 0$. Essentially, the minimal detectable group size is of order $N^{1/2}\gamma^2$: below this threshold, grouped heterogeneity is undetectable.

\subsection{Panels with increasing number of groups}\label{sec:Plarge}
The main message of our analysis remains valid even if the number of groups is allowed to grow with $N$. Consider
\begin{eqnarray*}\label{eq:P_alternative}
  H_{2,n}: \beta_{g(i)} =
  \begin{cases}
    \beta,              & \mbox{if }  g(i)= G_1,                       \\
    \beta+\lambda_p/\gamma, & \mbox{if } g(i)= G_p \mbox{ and } p=2,...,P,
  \end{cases}
\end{eqnarray*}
where $\{G_1,...,G_P\}$ forms a partition of $\mathcal{N}$, with $\lambda_p\neq\lambda_q$ for any $p\neq q\in\{2,...,P\}$. Let $M_{p} = \#G_p$ for each $p=1,...,P$, and $\bar{M} = \sum_{p=2}^P M_{p}$, where $M_p=M_p(N)$ and $P=P(N)$ may either remain fixed or diverge with $N$. We impose the restriction, $\sup_{N,p\le P}\|\lambda_p\|<\infty$.

Analyzing local power when $P\to\infty$ requires accounting for both the number of groups and their sizes. To do that, we assume an additional aggregate convergence condition for group-specific moment matrices.
\begin{assumption}\label{Ass3}
For a nonempty group $G$, define the sample moment matrices: 
\begin{align*}
\widehat S^{(G)} &= (\#G\,T)^{-1}\sum_{i\in G}x_i'x_i,\\ 
\widehat Q^{(G)} &=(\#GT)^{-1}\sum_{i\in G}x_i'x_i/\sigma_i^2,\\
\widehat W^{(G)} &= (\#GT)^{-1}\sum_{i\in G}x_i'(P_i-\operatorname{diag}\{h_i\})x_i,\\
\widehat\Omega_k^{(G)} &= (\#GT^2)^{-1}\sum_{i\in G}\sum_{t=2}^{T}\sum_{s=1}^{t-1}x_{it,k}x_{is,k}(x_{it}x_{is}'+x_{is}x_{it}')/2.
\end{align*}
For each of the matrices $\hat A^{(G_p)}\in\{\hat S^{(G_p)}, \hat Q^{(G_p)},\hat W^{(G_p)},\hat \Omega_1^{(G_p)},\ldots,\hat \Omega_K^{(G_p)}\}$, and the corresponding deterministic group matrices $A^{(G_p)}\in\{S^{(G_p)},Q^{(G_p)},W^{(G_p)},\Omega_1^{(G_p)},\ldots,\Omega_K^{(G_p)}\}$ assume that 
\begin{equation*}\label{eq:CorAggregate}
 \|\widehat A^{(G_1)}-A^{(G_1)}\|
 +\sum_{p=2}^{P}m_p
       \|\widehat A^{(G_p)}-A^{(G_p)}\|
 \xrightarrow{p}0.
\end{equation*}
\end{assumption}

Both the number of groups and the number of individuals per group determine the expansion rates under which the test statistics exhibit nontrivial power.

\setcounter{theorem}{0}
\begin{corollary}\label{Cor}
Under $H_{2,n}$, suppose that Assumptions~\ref{Ass1}--\ref{Ass3} hold. Then the conclusions of Theorem~\ref{theorem:Tlarge} hold with $M$ replaced by $\bar{M}$ and the noncentrality parameters replaced by
\begin{align*}
\delta_{\Delta}
&=\lim_{N,T\to\infty}\frac{1}{\sqrt{2K}} \frac{\bar{M} T}{\sqrt N\gamma^2}
\Bigg[\sum_{p=2}^{P}m_p\lambda_p'Q^{(G_p)}\lambda_p
 -\bar m\sum_{p,q=2}^{P}m_pm_q
 \lambda_p'Q^{(G_p)}Q^{-1}Q^{(G_q)}\lambda_q\Bigg],\\
\delta_{\mathcal J}
&=\lim_{N,T\to\infty}\frac{1}{\sqrt{\mathcal V_0}} \frac{\bar{M} T}{\sqrt N\gamma^2}
\Bigg[\sum_{p=2}^{P}m_p\lambda_p'W^{(G_p)}\lambda_p\\
&\qquad\qquad+\bar m\sum_{p,q=2}^{P}m_pm_q\lambda_p'
 \Bigl(S^{(G_p)}S^{-1}WS^{-1}S^{(G_q)}
       -2W^{(G_p)}S^{-1}S^{(G_q)}\Bigr)\lambda_q\Bigg],\\
\delta_{LM,k}
&=\lim_{N,T\to\infty}\frac{\bar{M} T}{\sqrt N\gamma^2}
\Bigg[\sum_{p=2}^{P}m_p\lambda_p'\Omega_k^{(G_p)}\lambda_p\\
&\qquad\qquad+\bar m\sum_{p,q=2}^{P}m_pm_q\lambda_p'
 \Bigl(S^{(G_p)}S^{-1}\Omega_kS^{-1}S^{(G_q)}
       -2S^{(G_p)}S^{-1}\Omega_k^{(G_q)}\Bigr)\lambda_q\Bigg],
\quad k=1,\ldots,K,
\end{align*}
where for all $p=2,...,P$, $m_p = M_{p}/\bar{M}$ and $\bar m = \bar{M}/N$.
\end{corollary}

Corollary~\ref{Cor} generalizes Theorem~\ref{theorem:Tlarge}. If $P$ is fixed and $\max_{2\le p\le P}M_p\to\infty$ as $N\to\infty$, then $\bar{M}\asymp\max_{2\le p\le P}M_p$, and the detection boundary has the same order as in Theorem~\ref{theorem:Tlarge} with $M$ replaced by $\bar{M}$. If $P\to\infty$ and the alternative group sizes are uniformly bounded, then $\bar{M}\asymp P$ and the boundary is $\gamma\asymp\sqrt{PT}/N^{1/4}$. More generally, when both $P \to \infty$ and $M_p \to \infty$ for at least one $p$, the rate is determined by the total ``volume'' of individuals in $\{G_2, \dots, G_P\}$, denoted as $\bar{M}$.

\section{Monte Carlo Study}\label{sec:mc}
We evaluate the finite-sample performance of the tests by comparing their power and assessing the theoretical results in the main text. The first subsection considers the three tests ($\Delta$, $\mathcal{J}$, and $LM$) under independent regression errors. The second investigates our conjecture in Section~2 that the detection boundary remains informative under serial correlation and cross-sectional dependence for suitably adapted tests.

Throughout, we use a data-generating process (DGP) adapted from PY to accommodate grouped slope patterns. For each $i=1,\ldots,N$
and $t=1,\ldots,T$, let
\begin{align*}
  y_{it} &= x_{it}'\beta_{g(i)} + \alpha_i + \epsilon_{it}, \\
  x_{it,k} &= \rho_{i,k} x_{it-1,k}
  + \sqrt{1-\rho_{i,k}^2}\,v_{it,k},
  \qquad k=1,\ldots,K,
\end{align*}
where $x_{it}=(x_{it,1},\ldots,x_{it,K})'$, $\alpha_i\overset{iid}{\sim}\mathrm{N}(1,1)$, $\rho_{i,k}\overset{iid}{\sim}\mathrm{U}(0.05,0.95)$, and, conditional on $\sigma_{ix}^2$, $v_{it,k}\sim\mathrm{N}(0,\sigma_{ix}^2)$ are independently distributed, with $\sigma_{ix}^2\sim 0.25+\chi^2_1$. The first 50 observations of $x_{it}$ ($t=-49,\ldots,0$) are discarded to mitigate the effect of initial values. 

Group-specific slopes, $\beta_{g(i)}=(\beta_{g(i),1},\ldots,\beta_{g(i),K})'$, are generated as
\[
  \beta_{g(i),k} =
  \begin{cases}
    1, & \text{if } i\in G_1, \\
    1+\lambda(1+e_p), & \text{if } i\in G_p,
  \end{cases}
\]
for $k=1,\ldots,K$, and $p=2,\ldots,P$, where $e_p\overset{iid}{\sim}\mathrm{U}(-h,h)$. The dominant group $G_1$ contains $N-\sum_{p=2}^{P}M_p$ individuals, while each alternative group $G_p$ contains $M_p$ individuals. The parameters $\lambda$ and $h$ control, respectively, the magnitude of slope differences relative to the dominant group and their dispersion across alternative groups. The following subsections differ in the specification of the regression errors $\epsilon_{it}$.

\subsection{Independent errors}
Let $\epsilon_{it}\overset{iid}{\sim}\mathrm{N}(0,1)$. We consider two designs to assess how group size and the magnitude of slope heterogeneity affect power:
\begin{itemize}
  \item[(1)] \textbf{Two-group design:} $P=2$, $h=0$, varying the size of the alternative group $M_2=1,...,N/2$ for $\lambda\in\{0.1,0.2,0.3\}$.
  \item[(2)] \textbf{Multiple-group design:} $P\in\{3,5,10\}$, where heterogeneous individuals are randomly assigned to $P-1$ groups of approximately equal size. Here, $\lambda=0.2$ and $h=0.2$, focusing on the effect of the \emph{overall mass of the alternative groups} $\bar{M}$. When $\bar{M}<P-1$, each heterogeneous individual forms a separate group.
\end{itemize}
Each design is run at $(N,T) = (100,10)$ and $(100,100)$, corresponding to short and long panels, and at $K=1$ and $K=4$, corresponding to single and multiple regression. The number of replications is set to $R_1=1000$.

In both designs, the $\mathcal{J}$ test is severely oversized when $T$ is small, especially for $K=4$, and its rejection rate far exceeds that of the other two tests. The simulated rejection rates under slope homogeneity are reported in Table \ref{Table1} below. To compare the three tests on equal footing, we use size-adjusted critical values, obtained by simulating the null distribution of each test statistic using the same DGP as in the power experiments, but imposing slope homogeneity, $\beta_{g(i),k}=1$ for all $i=1,...,N$ and $k=1,...,K$. For each combination of $(K,T)$, we generate $R_2=10,000$ independent samples under $H_0$. The size-adjusted critical values are the empirical 95th percentiles of the resulting test statistics, using their absolute values for the two-sided tests.

\begin{center}
\begin{minipage}{0.85\textwidth}
\centering
\begin{tabular}{ccccc}
\toprule
$K$ & $T$ & $LM$ & $\Delta$ & $\mathcal{J}$ \\
\midrule
1 & 10  & 5.02 & 4.61 & 12.94 \\
1 & 100 & 5.53 & 4.73 & 5.48 \\
\midrule
4 & 10  & 4.40 & 5.38 & 54.60 \\
4 & 100 & 5.57 & 5.26 & 5.28 \\
\bottomrule
\end{tabular}
\captionof{table}{Empirical Size of the Tests}
\label{Table1}
\vspace{2pt}
\begin{minipage}{\textwidth}
\footnotesize
\raggedright
Empirical rejection frequencies under $H_0$ based on 10{,}000 Monte Carlo replications, using asymptotic critical values at the 5\% nominal significance level. Entries are percentages.
\end{minipage}
\end{minipage}
\end{center}

The DGP includes individual effects, which are removed through appropriate demeaning procedures for each test statistic. For the $LM$ test, we employ orthogonal deviations as in \cite{arellano1995another}, following Remark~5.1 in BRS. For the $\Delta$ and $\mathcal{J}$ tests, we apply the standard within transformation, pre-multiplying $y_i$ and $x_i$ by $P^{\perp}_{\iota}=I_T-T^{-1}\iota_T\iota_T'$, where $\iota_T$ is a $T\times1$ vector of ones, as suggested in PY and \cite{lu2017determining}.

\begin{figure}
  \centering
  \textbf{\large Two Groups: $K=1$, $N=100$, $T=10$}
  \par\smallskip
  \includegraphics[scale=0.50]{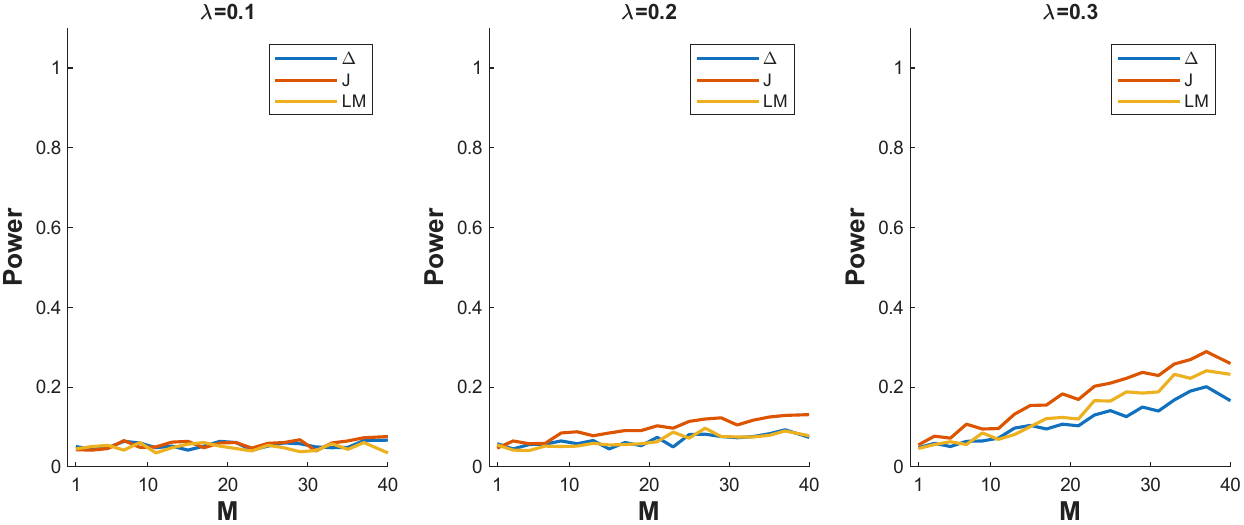}\par\medskip
  \textbf{\large Two Groups: $K=1$, $N=100$, $T=100$}
  \par\smallskip
  \includegraphics[scale=0.50]{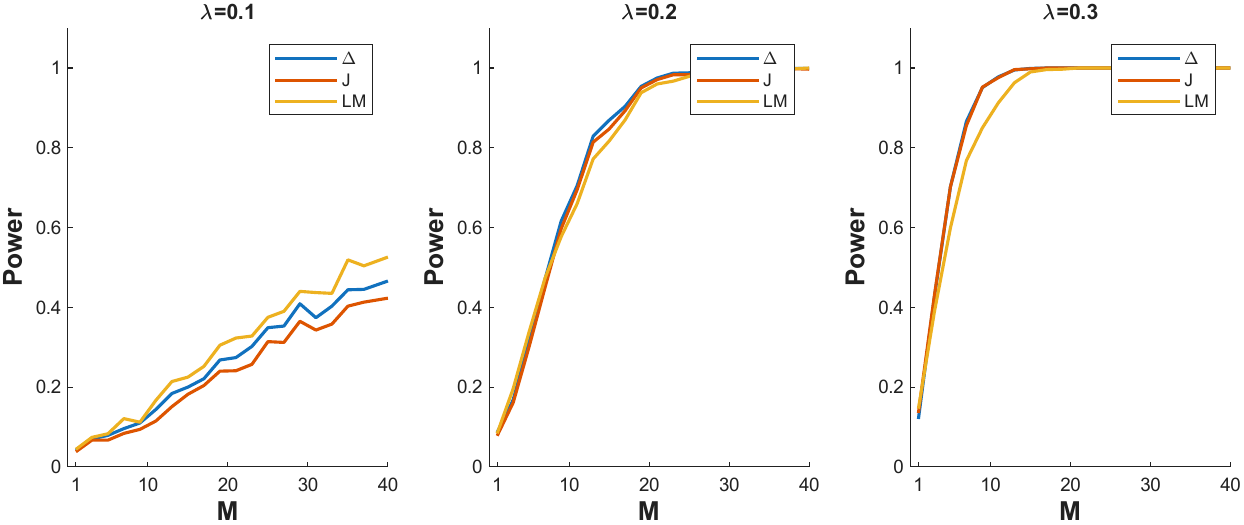}\par\medskip
  \textbf{\large Two Groups: $K=4$, $N=100$, $T=10$}
  \par\smallskip
  \includegraphics[scale=0.50]{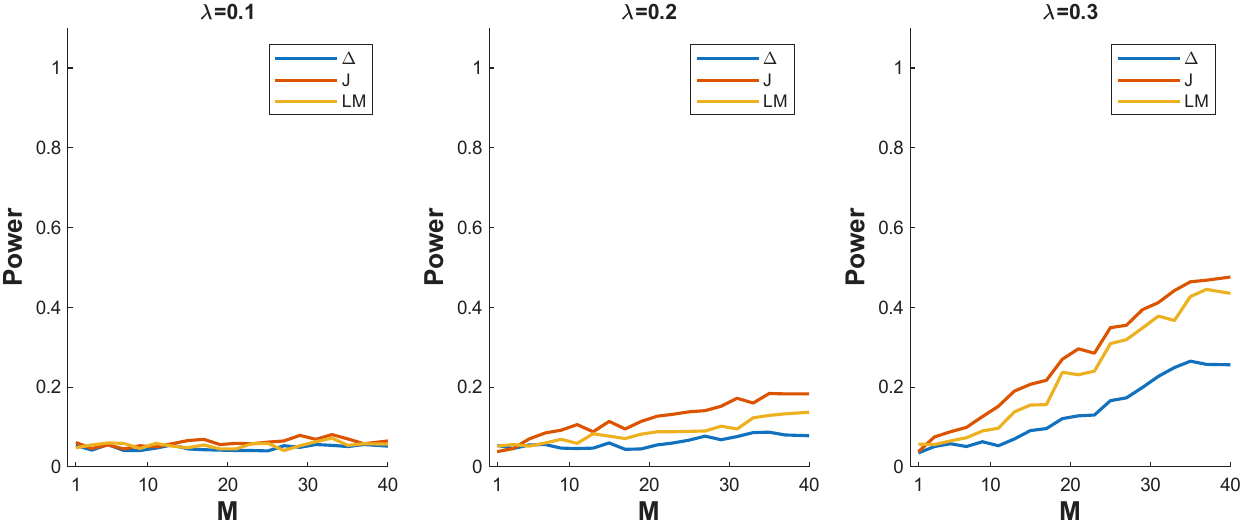}\par\medskip
  \textbf{\large Two Groups: $K=4$, $N=100$, $T=100$}
  \par\smallskip
  \includegraphics[scale=0.50]{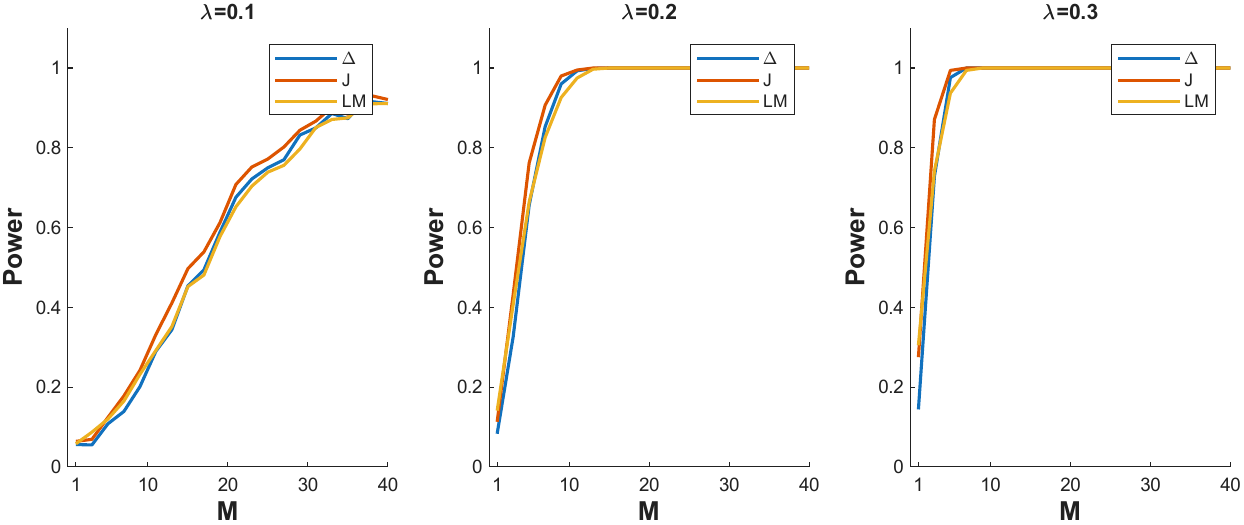}
  \caption{Power curves of slope homogeneity tests with two groups under independent errors.}
  \label{fig:Figure1}
\end{figure}

\begin{figure}
  \centering
  \textbf{\large $P$ Groups: $K=1$, $N=100$, $T=10$}
  \par\smallskip
  \includegraphics[scale=0.50]{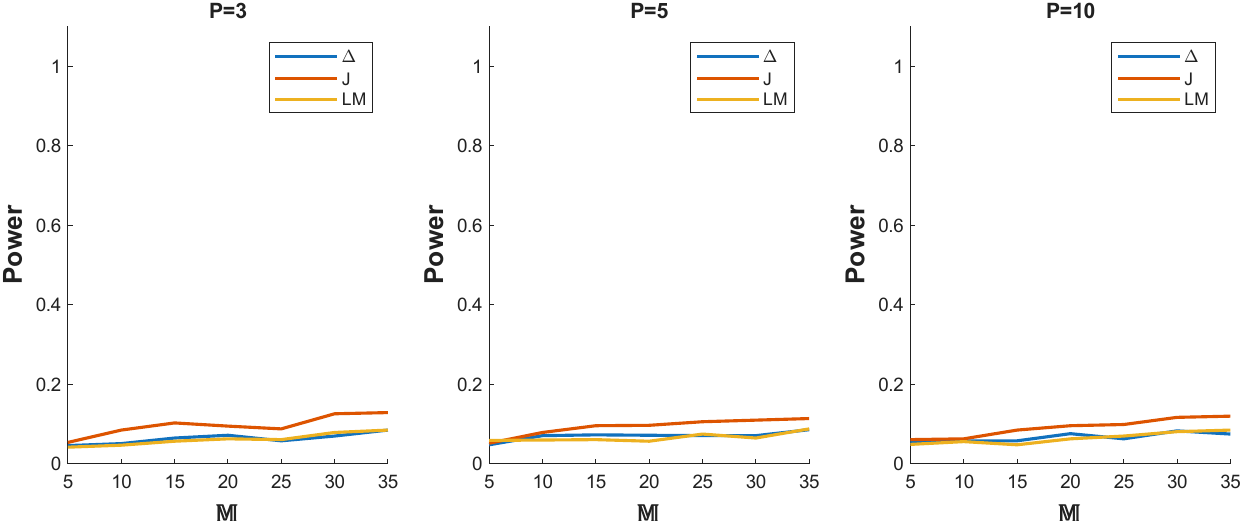}\par\medskip
  \textbf{\large $P$ Groups: $K=1$, $N=100$, $T=100$}
  \par\smallskip
  \includegraphics[scale=0.50]{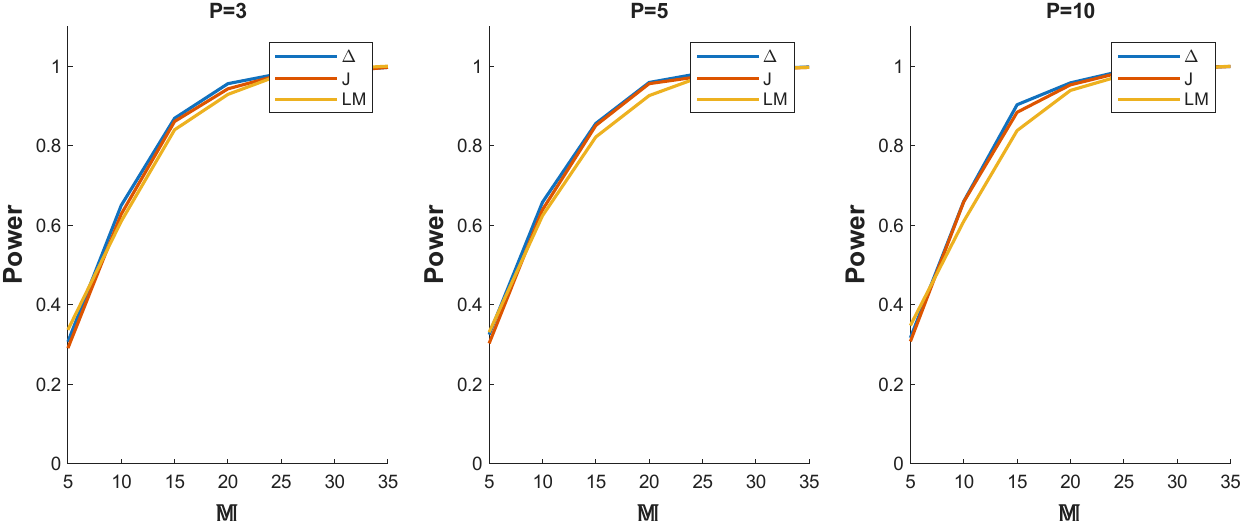}\par\medskip
  \textbf{\large $P$ Groups: $K=4$, $N=100$, $T=10$}
  \par\smallskip
  \includegraphics[scale=0.50]{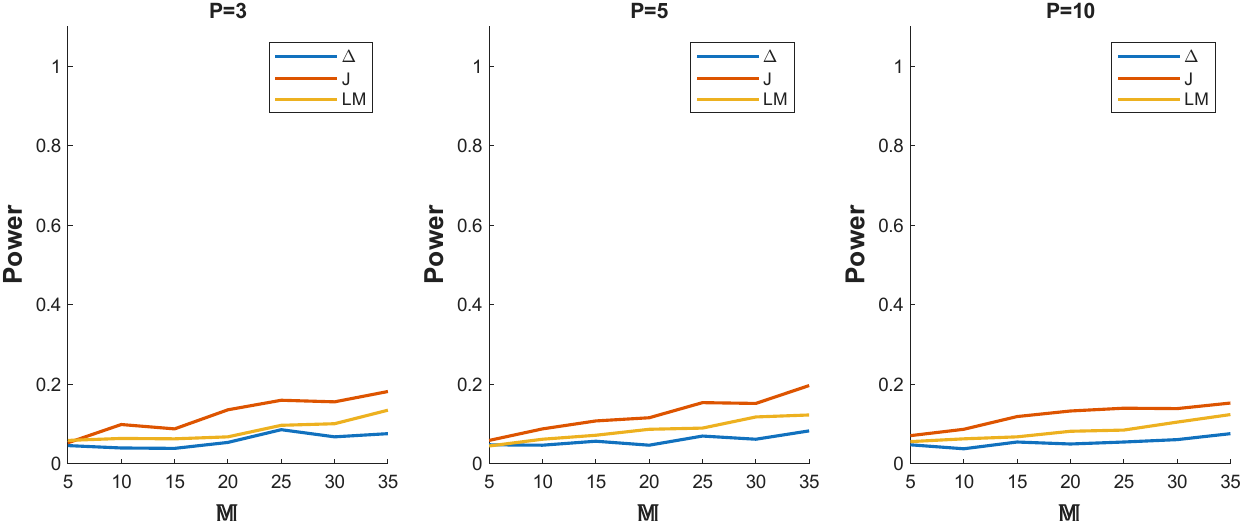}\par\medskip
  \textbf{\large $P$ Groups: $K=4$, $N=100$, $T=100$}
  \par\smallskip
  \includegraphics[scale=0.50]{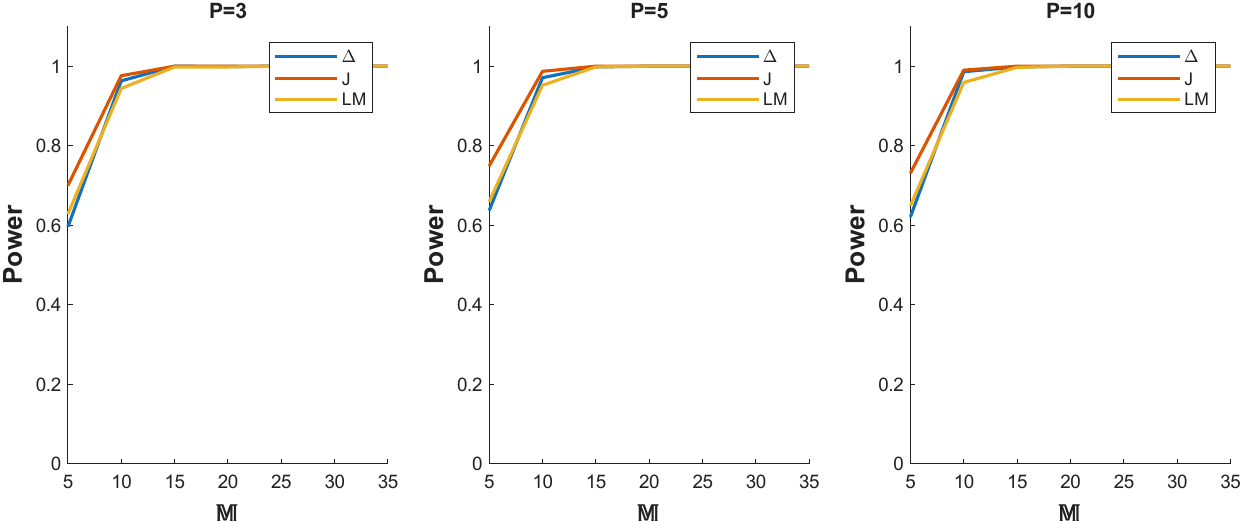}
  \caption{Power curves of slope homogeneity tests with $P$ groups under independent errors.}
  \label{fig:Figure2}
\end{figure}
All simulations report rejection frequencies at the 5\% significance level based on 1{,}000 Monte Carlo replications. Results for the two-group design are shown in Figure~\ref{fig:Figure1}, and those for the multiple-group design in Figure~\ref{fig:Figure2}. 

The two-group results closely mirror our theoretical findings in Theorem~\ref{theorem:Tlarge}. Rejection rates increase monotonically with the alternative group size $M$ and with the slope deviation parameter $\lambda$. Power is limited for small alternative groups and increases as either $M$ or $T$ grows, illustrating the joint role of group size and panel length in detectability. In the multiple-group experiment, the decisive role is played by the number of individuals in the alternative groups, $\bar{M}$. Tests exhibit substantial power when $\bar{M}$ is larger than $\sqrt{N}$, consistent with Corollary~\ref{Cor}. Power declines slightly as $P$ increases, but this effect is secondary once $\bar{M}$ or $T$ is sufficiently large. 

In most cases, all three tests show comparable power. The $J$ and $LM$ tests, however, tend to gain power somewhat more rapidly than the $\Delta$ test, especially in the short-panel designs with $T=10$ and $K=4$. This divergence is much less visible once $T$ increases. When $T=100$, the three procedures behave in a way that is broadly consistent with the asymptotic comparison in Section~4.

\subsection{Serial correlation and cross-sectional dependence}
To assess the conjecture in Section~2, we retain the baseline DGP, changing only the error structure. The first DGP considers serially correlated errors,
\[
    \epsilon_{it}=\phi_i\epsilon_{it-1}
    +\sqrt{1-\phi_i^2}\,\eta_{it},
\]
where $\phi_i\overset{iid}{\sim}\mathrm{U}(0,0.5)$ and $\eta_{it}\overset{iid}{\sim}\mathrm{N}(0,1)$. The second introduce cross-sectional dependence through a common-factor structure,
\[
    \epsilon_{it}=\ell_i f_t+e_{it},
\]
where $f_t$ and $e_{it}$ are independent standard normal random variables and $\ell_i\overset{iid}{\sim}\mathrm{N}(1,1)$. For both error structures, we restrict attention to the two-group design of the previous subsection, with $N=100$, $T\in\{10,100\}$, and $K=1$.

For the first DGP, we consider the heteroskedasticity and autocorrelation consistent (HAC) extension of the PY test proposed by \citet{blomquist2013testing}, which retains the structure of $\Delta$ in \eqref{eq:delta}, replacing the individual error variance estimator with a HAC covariance estimator. For $K=1$, the test statistic is
\[
\Delta_{\mathrm{HAC}} = \frac{\mathcal{S}(\hat{V}_i)-N}{\sqrt{2N}},
\]
where $\hat{V}_i$ is an HAC estimator of the long-run variance of $T\tilde{x}_{it}\tilde{\epsilon}_{it}/(\tilde x_i'\tilde x_i)$, with $\tilde{\epsilon}_i=\tilde{y}_i-\tilde{x}_i\hat{\beta}_{WLS}(\hat{V}_i)$. Here, $\tilde{x}_{it}= P^{\perp}_{\iota}x_{it}$, and $\tilde{y}_{it}=P^{\perp}_{\iota}y_{it}$ denote the time-demeaned data. The functions $\mathcal{S}(\cdot)$ and $\hat{\beta}_{WLS}(\cdot)$ are defined as in Section~3, using the time-demeaned data.

For the second, we use the original SC test, which accommodates cross-sectional dependence through an interactive factor structure. The original factor-based SC test retains the form in \eqref{eq:SC}, with restricted residuals $\tilde{\epsilon}_{F,i}=\tilde{y}_i-\tilde{x}_i\hat{\beta}_{LS} - \hat{F}\hat{\ell}_i$, where $\hat{F}$, and $\hat{\ell}_i$ are the restricted estimators of factor and factor loading, respectively, and $\hat{\beta}_{LS}$ uses time-demeaned data. The $LM$ statistic becomes
\[
LM_{SC}^{F} =\sum_{i=1}^N \tilde{\epsilon}_{F,i}'\tilde{x}_i (\tilde{x}_i'\tilde{x}_i)^{-1}\tilde{x}_i'\tilde{\epsilon}_{F,i}.
\]
The bias and variance estimators are defined as in Section~3, replacing $\hat{\epsilon}_i$ with $\tilde{\epsilon}_{F,i}$ and using
\[
P_i = P_{\scriptscriptstyle \hat F}^{\perp}\tilde{x}_i(\tilde{x}_i'\tilde{x}_i)^{-1}\tilde{x}_i'P_{\scriptscriptstyle \hat F}^{\perp},
\]
where $P_{\scriptscriptstyle \hat F}^{\perp}=I_T-\hat{F}(\hat{F}'\hat{F})^{-1}\hat{F}'$. This corresponds to the time-demeaned formulation in Remark 7 of SC.

The HAC covariance estimator used in the HAC-PY test employs AR(1) prewhitening (see Remark 4 of \citealt{blomquist2013testing}), a quadratic spectral kernel, and the individual-specific, data-driven bandwidth of \citet{andrews1991heteroskedasticity}. For the SC test, the restricted residuals are obtained by jointly estimating the common slope and a single interactive factor under the null, using an iterative principal-components procedure after removing individual fixed effects. The factor loadings are estimated by least squares conditional on the estimated slope and factor, following Section 2.2 of SC. 

Both the $\Delta_{HAC}$ and SC tests exhibit substantial size distortions when $T$ is small. Table~\ref{Table2} reports their empirical sizes using asymptotic critical values. To ensure comparability, we therefore report size-adjusted rejection frequencies, using critical values obtained from $R_2=10{,}000$ Monte Carlo replications under the null of
slope homogeneity. These critical values are computed separately for each error structure, using the corresponding null DGP.

\begin{center}
\begin{minipage}{0.75\textwidth}
\centering
\begin{tabular}{cccccc}
\toprule
\multicolumn{3}{c}{AR(1) errors} &
\multicolumn{3}{c}{Factor errors} \\
\cmidrule(lr){1-3}\cmidrule(lr){4-6}
$N$ & $T$ & $\Delta_{\mathrm{HAC}}$ &
$N$ & $T$ & $\mathcal{J}$ \\
\cmidrule(lr){1-3}\cmidrule(lr){4-6}
100 & 10  & 56.69 & 100 & 10  & 48.86 \\
100 & 100 & 3.46  & 100 & 100 & 5.62  \\
\bottomrule
\end{tabular}

\captionof{table}{Empirical Size under Serial Correlation
and Cross-Sectional Dependence}
\label{Table2}
\vspace{2pt}

\begin{minipage}{\textwidth}
\footnotesize
\raggedright
Empirical rejection frequencies under $H_0$ with $K=1$, based on 10,000 Monte Carlo replications and asymptotic critical values. The nominal significance level is 5\%. Entries are percentages.
\end{minipage}
\end{minipage}
\end{center}

The rejection rates of the $\Delta_{HAC}$ test under AR(1) errors are slightly lower than those of the conventional $\Delta$ test under independent errors. The $\mathcal{J}$ test appears robust to cross-sectional dependence, with rejection patterns close to those obtained under cross-sectional independence. Overall, the size-adjusted power curves under both error structures resemble those of the baseline two-group experiment, providing numerical evidence that the detection boundary remains informative for grouped slope heterogeneity in these more general settings.

\section{Conclusions}
We derive the detection boundary of standard slope-homogeneity tests under small-group heterogeneity. The power of the tests is evaluated under a sequence of doubly local alternatives, linking the degree of heterogeneity across groups and their sizes to test performance. Our analysis identifies the conditions under which slope-homogeneity tests can detect such departures, including those confined to groups whose cardinality is negligible relative to the cross-sectional size. We also characterize when small-group heterogeneity remains undetectable by these tests yet sufficiently consequential to invalidate inference based on pooled slope estimates. Monte Carlo simulations support the theoretical findings in finite samples.

\begin{figure}[tbp]
  \centering
  \textbf{\large AR(1), $T=10$}\par\smallskip
  \includegraphics[scale=0.50]{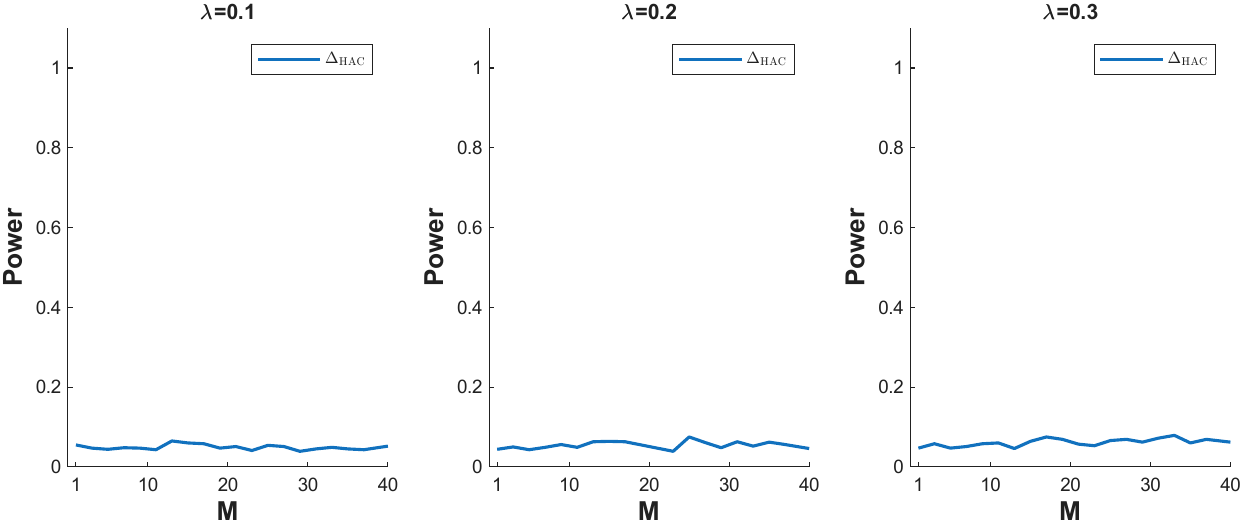}\par\medskip
  \textbf{\large AR(1), $T=100$}\par\smallskip
  \includegraphics[scale=0.50]{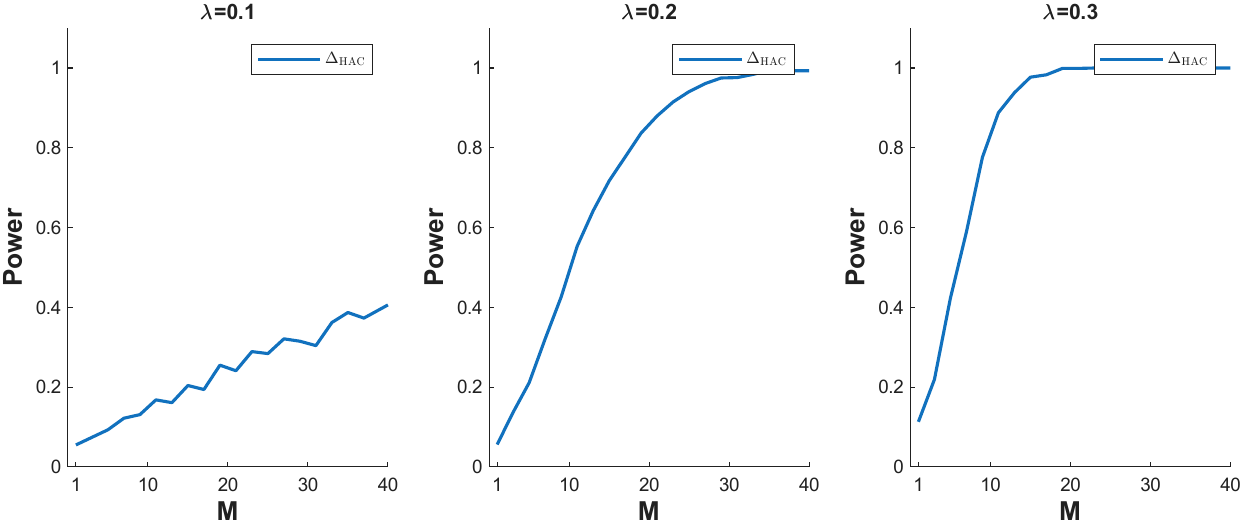}\par\medskip
  \textbf{\large Cross-Sectional Dependence, $T=10$}\par\smallskip
  \includegraphics[scale=0.50]{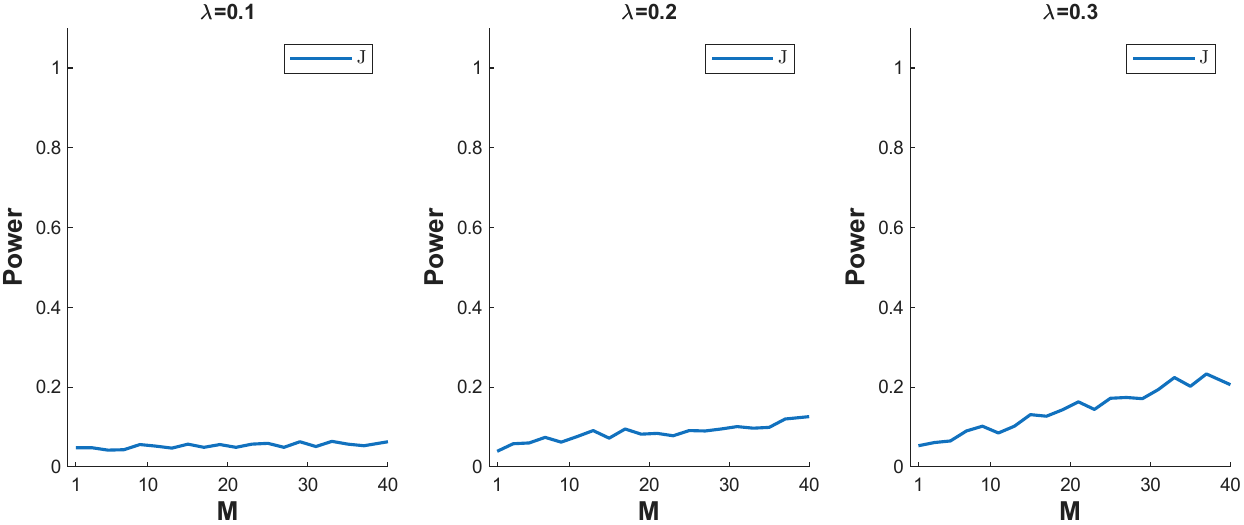}\par\medskip
  \textbf{\large Cross-Sectional Dependence,, $T=100$}\par\smallskip
  \includegraphics[scale=0.50]{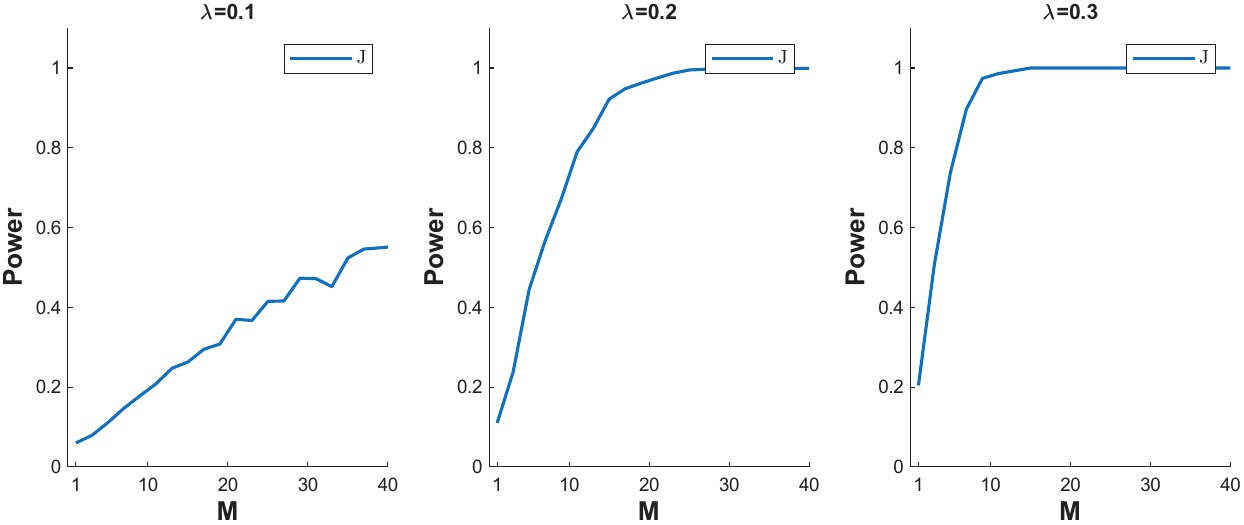}
  \caption{Power curves of slope homogeneity tests with $2$ groups under serially correlated errors and cross-sectional dependence.}
  \label{fig:Figure3}
\end{figure}
\clearpage

\newpage

\appendix
\setcounter{equation}{0}
\renewcommand{\theequation}{A\arabic{equation}}
\renewcommand{\thetable}{A\arabic{table}}
\renewcommand{\thelemma}{\arabic{lemma}}
\setcounter{table}{0}
\section{Proofs}\label{sec:proofs}

We derive the asymptotic distributions of the three test statistics under $H_{1,n}$ as $N,T\to\infty$ (Theorem~\ref{theorem:Tlarge}). The rates in Corollary~1 follow by setting $M=\bar M$ and invoking the uniform convergence condition in Assumption~\ref{Ass3}. The derivation of the corresponding noncentrality parameters is
discussed at the end of this section. The next section establishes the fixed-$T$ asymptotics of the BRS test (Corollary~2).

The PY test requires additional regularity conditions.
For some $\eta>0$ and $C<\infty$,
\begin{equation}\label{eq:PYregularity}
    \begin{gathered}
        \E{\left\{
            \frac{\epsilon_i'(I_T-P_i)\epsilon_i}{T-1}
        \right\}^{-(4+\eta)}\Biggm| x_i}\leq C
        \quad\text{a.s.},\quad 
        \frac{1}{N}\sum_{i=1}^N\sum_{t=1}^T
        P_{i,tt}^2\xrightarrow{p}0,
    \end{gathered}
\end{equation}
where the moment bound holds uniformly in $i$, for all sufficiently large $T$.

The proofs follow a common strategy: we decompose each test statistic into a stochastic term that determines its limiting distribution under the null, a term that drives its local power, and remainder terms, denoted by $R_{NT}$ and $C_{NT}$, whose asymptotic relevance depends on the sequence of alternatives.

Under $H_{1,n}$, we can rewrite model \eqref{eq:regmodel} as
\begin{eqnarray*}
  y_{it} = x_{it}^{\prime}\beta_i + \epsilon_{it}= x_{it}^{\prime}\beta + u_{it},
\end{eqnarray*}
where $u_{it}=x_{it}^{\prime}\lambda_i+\epsilon_{it}$, $\lambda_{i} =0$ if $i\in G_1$ and $\lambda_{i} =\lambda/\gamma$ if $ i\in G_2$.

\begin{lemma}\label{lem:Lemma1}
  Under $H_{1,n}$ and Assumptions \ref{Ass1}, \ref{Ass2},
  \begin{equation*}
    \hat{\beta}_{LS}-\beta= O_p\left(\frac{M}{N\gamma} + \frac{1}{\sqrt{NT}}\right), \, \mbox{as } N,T\to\infty.
  \end{equation*}
  \begin{proof}[Proof of Lemma 1]
    \begin{eqnarray*}
      \hat{\beta}_{LS}-\beta &=& \left(\sum_{i=1}^N\sum_{t=1}^Tx_{it}x_{it}'\right)^{-1}\sum_{i=1}^N\sum_{t=1}^Tx_{it}\epsilon_{it} + \left(\sum_{i=1}^N\sum_{t=1}^Tx_{it}x_{it}'\right)^{-1}\sum_{i=1}^N\sum_{t=1}^Tx_{it}x_{it}^{\prime}\lambda_i.
    \end{eqnarray*}
    From Assumption \ref{Ass2} we have
    \begin{eqnarray*}
        \frac{1}{NT}\sum_{i=1}^N\sum_{t=1}^Tx_{it}x_{it}'& \overset{p}{\to}& S \text{ as } N,T\to\infty
    \end{eqnarray*}
    Further, by Lemma 2 in PY, $(NT)^{-1}\sum_{i=1}^N\sum_{t=1}^Tx_{it}\epsilon_{it}=O_p\left((NT)^{-1/2}\right)$ as $N,T\to\infty$.
    Finally, from Assumption \ref{Ass2}
    \begin{equation*}
      \sum_{i=1}^N\sum_{t=1}^Tx_{it}x_{it}^{\prime}\lambda_i = \frac{1}{\gamma}\sum_{i\in G_2}\sum_{t=1}^T x_{it}x_{it}^{\prime}{\lambda}=
        O_p\left(\frac{MT}{\gamma}\right),  \mbox{as } N,T\to\infty.
    \end{equation*}
  \end{proof}
\end{lemma}

\begin{proof}[Proof Theorem 1, BRS Test]\\
  We start by considering the $LM$ test of BRS under $H_{1,n}$. After substituting $\hat{\epsilon}_{it}=u_{it}-x_{it}^{\prime}(\hat\beta_{LS}-\beta)$, with $u_{it}=x_{it}^{\prime}\lambda_i+\epsilon_{it}$, into the score vector, we obtain the following decomposition,
  \begin{equation}\label{eq:decomposition}
    \hat{s}_k = \sum_{i=1}^{N}\sum\limits_{t=2}^{T}\sum\limits_{s=1}^{t-1}\epsilon_{it}\epsilon_{is}x_{it,k}x_{is,k} + \sum_{i=1}^{N}\sum\limits_{t=2}^{T}\sum\limits_{s=1}^{t-1}\lambda_i'{x}_{it}x_{it,k}x_{is,k}{x}_{is}'\lambda_i + R_{NT,k} + C_{NT,k},
  \end{equation}
  where the first two terms drive the asymptotic normality and the power of the test, respectively, while $R_{NT,k}$ and $C_{NT,k}$ are remainder terms given by,
  \begin{eqnarray*}
  R_{NT,k} &=&\sum_{i=1}^N\sum_{t=2}^{T}\sum_{s=1}^{t-1} (\hat\beta_{LS}-\beta)'x_{it}x_{it,k}x_{is,k}x_{is}' (\hat\beta_{LS}-\beta)\\
           &&\quad-\sum_{i=1}^N\sum_{t=2}^{T}\sum_{s=1}^{t-1}x_{it,k}x_{is,k} \left\{(\hat\beta_{LS}-\beta)'x_{it}u_{is}+(\hat\beta_{LS}-\beta)'x_{is}u_{it}\right\},\\
    C_{NT,k}&=& \sum_{i=1}^{N}\sum\limits_{t=2}^{T}\sum\limits_{s=1}^{t-1}\lambda_i'{x}_{it}\epsilon_{is}x_{it,k}x_{is,k}+ \sum_{i=1}^{N}\sum\limits_{t=2}^{T}\sum\limits_{s=1}^{t-1}\epsilon_{it}\lambda_i'{x}_{is}x_{it,k}x_{is,k},\\
  \end{eqnarray*}
  for $k=1,...,K$.
  \noindent After substituting $u_{it}=x_{it}^{\prime}\lambda_i+\epsilon_{it}$ into $R_{NT}$ and
  using Lemma~\ref{lem:Lemma1}, Lemma~A2 in BRS, together with the fact that $\lambda_i$ is non zero for only $M$ elements, we obtain,
  \begin{align*}
    \frac{1}{\sqrt{N}T}R_{NT,k} = O_p\left(\frac{M^2T}{\gamma^2{N}^{3/2}}\right)+o_p(1).
  \end{align*}
  Furthermore, if $M/N\to m_0>0$ and $MT/(\sqrt N\gamma^2)\to c^2$, then
  \begin{align*}
    \frac{1}{\sqrt{N}T}R_{NT,k} & = c^2m_0\lambda'\left(S^{(G_2)}S^{-1}\Omega_kS^{-1}S^{(G_2)} - 2S^{(G_2)}S^{-1}\Omega_{k}^{(G_2)}\right)\lambda + o_p(1).
  \end{align*}
  By the same arguments,
  \[
    \frac{1}{\sqrt{N}T}C_{NT,k} = O_p\left(\frac{\sqrt{MT}}{\sqrt{N}\gamma}\right),
  \]
  and if $MT/(\sqrt{N}\gamma^2)\to c^2$ then $\sqrt{MT}/(\sqrt{N}\gamma)\to 0$ resulting in $C_{NT,k}/(\sqrt{N}T) = o_p(1)$.
  Further,
  \begin{align*}
  \frac{1}{\sqrt{N}T}\sum\limits_{i=1}^{N}\sum\limits_{t=2}^{T}\sum\limits_{s=1}^{t-1}\lambda_i'{x}_{it}x_{it,k}x_{is,k}{x}_{is}'\lambda_i &=\frac{1}{\gamma^2\sqrt{N}T} \sum_{i\in G_2}\sum\limits_{t=2}^{T}\sum\limits_{s=1}^{t-1}\lambda'{x}_{it}x_{it,k}x_{is,k}{x}_{is}'\lambda\\ &\overset{p}{\rightarrow}\lambda' \Omega_k^{(G_2)}\lambda\times \lim\limits_{N,T\to\infty} \frac{MT}{\sqrt{N}\gamma^2}.
  \end{align*}
  Let $s_k:=\sum_{i=1}^{N}\sum_{t=2}^{T}\sum_{s=1}^{t-1}\epsilon_{it}\epsilon_{is}x_{it,k}x_{is,k}$ and $s=[s_1,...,s_K]'$ then by Theorems 5.1 and Corollary 5.2 in BRS,
  \begin{eqnarray*}
    \frac{1}{\sqrt{N}T}s \overset{d}{\rightarrow}N(0,\Psi)  \text{ as } N,T\to\infty.
  \end{eqnarray*}
Finally, we establish the consistency of $\hat V$, following the argument in Remark~6.2 of BRS. Let $ v_{it,k}=\epsilon_{it}x_{it,k}\sum_{s<t}\epsilon_{is}x_{is,k}$ and $\hat v_{it,k}=\hat\epsilon_{it}x_{it,k} \sum_{s<t}\hat\epsilon_{is}x_{is,k}$, such that $$\hat V_{k,l}=\sum_{i,t}\hat v_{it,k}\hat v_{it,l},$$ where sums over $t$ run from $2$ to $T$. Under Assumptions~\ref{Ass1}--\ref{Ass2}, and using Minkowski's and H\"older's inequalities, we have
\begin{equation*}
 \sup_{i,T}E\left|T^{-2}\sum_t v_{it,k}v_{it,l}\right|^{3/2}
 <\infty.
\end{equation*}
Applying the law of large numbers across individuals yields
\begin{equation*}
 \frac1{NT^2}\sum_{i,t}v_{it,k}v_{it,l}
 =\frac1{NT^2}\sum_i\sum_{t=2}^{T}\sum_{s=1}^{t-1}
       \sigma_i^4x_{it,k}x_{it,l}x_{is,k}x_{is,l}+o_p(1)
 =\Psi_{k,l}+o_p(1).
\end{equation*}
It remains to justify replacing the true errors by residuals: $ (NT^2)^{-1}\hat V_{k,l} = (NT^2)^{-1}\sum_{i,t}v_{it,k}v_{it,l}+o_p(1)$. Consider the following decomposition of the difference
\begin{align*}
 \hat v_{it,k}-v_{it,k}
 =x_{it,k}\Bigg(\epsilon_{it}\sum_{s<t}(\hat\epsilon_{is}-\epsilon_{is})x_{is,k}+(\hat\epsilon_{it}-\epsilon_{it})
                           \sum_{s<t}\epsilon_{is}x_{is,k}+(\hat\epsilon_{it}-\epsilon_{it})\sum_{s<t}(\hat\epsilon_{is}-\epsilon_{is})x_{is,k}\Bigg),
\end{align*}
and define $D_{NT,k}=(NT^2)^{-1}\sum_{i,t}(\hat v_{it,k}-v_{it,k})^2$. Let
\[
d_{NT}=\frac{M}{N\gamma} + \frac{1}{\sqrt{NT}}\qquad a_{NT} = \frac{MT}{\sqrt{N}\gamma^2}.
\]
Lemma~\ref{lem:Lemma1} gives $\|\hat\beta_{LS}-\beta\|=O_p(d_{NT})$. Applying moment bounds to the preceding decomposition, and using the sixth-moment bounds in Assumptions~\ref{Ass1}--\ref{Ass2} together with
Lemma~\ref{lem:Lemma1}, we obtain
\begin{align*}
 D_{NT,k}
 &=O_p\left[
       T\left\{\frac{M}{N\gamma^2}+d_{NT}^2\right\}
       +\frac{TM^{4/3}}{N\gamma^4}
       +TN^{1/3}d_{NT}^4\right]\\
 &=O_p\left(\frac{a_{NT}}{\sqrt N}
             +\frac1N+\frac{a_{NT}^2}{M^{2/3}T}\right)
 =o_p(1).
\end{align*}
The last rate uses $M\le N$ and $(a+b)^4\le8(a^4+b^4)$, by the $C_r$ inequality; it also covers fixed $M\ge1$. Finally, expand
\begin{equation*}
 \hat v_{it,k}\hat v_{it,l}-v_{it,k}v_{it,l}
 =(\hat v_{it,k}-v_{it,k})v_{it,l}
       +v_{it,k}(\hat v_{it,l}-v_{it,l})+(\hat v_{it,k}-v_{it,k})(\hat v_{it,l}-v_{it,l}).
\end{equation*}
Cauchy--Schwarz applied to each of the three sums, together with
$D_{NT,k}=D_{NT,l}=o_p(1)$, gives
\begin{equation*}
 \frac{\hat V_{k,l}}{NT^2}
 =\frac1{NT^2}\sum_{i,t}v_{it,k}v_{it,l}+o_p(1).
\end{equation*}
Combining these results in \eqref{eq:LMtest} establishes
the corresponding statement of the theorem.

\end{proof}

\begin{proof}[Proof Theorem 1, PY Test]\\
We now use the same procedure to prove the statement for the PY test. Under $H_{1,n}$, we have that,
  \[
    \left(\hat{\beta}_{i}-\hat{\beta}_{WLS}(\tilde{\sigma}_i^2)\right) = \left[\lambda_i + \left(\frac{x_i'x_i}{\tilde{\sigma}^2_i}\right)^{-1}\left(\frac{x_i'\epsilon_i}{\tilde{\sigma}^2_i}\right) - \left(\sum\limits_{i=1}^N\frac{x_i'x_i}{\tilde{\sigma}^2_i}\right)^{-1}\left(\sum\limits_{i=1}^N\frac{x_i'u_i}{\tilde{\sigma}^2_i}\right)\right]
  \]
  where $u_i=(u_{i1},...,u_{iT})' = x_i\lambda_i + \epsilon_i$. Thus, rewrite $S_{PY}$ as,
  \[
    S_{PY} = \sum_{i=1}^N \frac{\epsilon_i'P_i\epsilon_i}{\tilde{\sigma}^2_i} + \sum\limits_{i=1}^N\frac{\lambda_i'x_i'x_i\lambda_i}{\tilde{\sigma}^2_i} + R_{NT} + C_{NT},
  \]
  where 
  \begin{eqnarray*}
    R_{NT} &=& - \left(\sum\limits_{i=1}^N\frac{\lambda_i'x_i'x_i}{\tilde{\sigma}^2_i}\right)\left(\sum\limits_{i=1}^N\frac{x_i'x_i}{\tilde{\sigma}^2_i}\right)^{-1}\left(\sum\limits_{i=1}^N\frac{x_i'x_i\lambda_i}{\tilde{\sigma}^2_i}\right) - \left(\sum\limits_{i=1}^N\frac{\epsilon_i'x_i}{\tilde{\sigma}^2_i}\right)\left(\sum\limits_{i=1}^N\frac{x_i'x_i}{\tilde{\sigma}^2_i}\right)^{-1}\left(\sum\limits_{i=1}^N\frac{x_i'\epsilon_i}{\tilde{\sigma}^2_i}\right),
\end{eqnarray*}
and
\begin{eqnarray*}
    C_{NT} &=& 2\sum_{i=1}^N \frac{\epsilon_i'x_i\lambda_i}{\tilde{\sigma}^2_i} - 2\left(\sum\limits_{i=1}^N\frac{\epsilon_i'x_i}{\tilde{\sigma}^2_i}\right)\left(\sum\limits_{i=1}^N\frac{x_i'x_i}{\tilde{\sigma}^2_i}\right)^{-1}\left(\sum\limits_{i=1}^N\frac{x_i'x_i\lambda_i}{\tilde{\sigma}^2_i}\right).
  \end{eqnarray*}
Let $\theta_i=\lambda_i-(\hat\beta_{LS}-\beta)$ such that
  \begin{equation*}
 \tilde\sigma_i^2  =\frac{\epsilon_i'\epsilon_i}{T-1}+\frac{2\epsilon_i'x_i\theta_i}{T-1} +\frac{\theta_i'x_i'x_i\theta_i}{T-1}.
  \end{equation*}
Put $q_i=\epsilon_i'\epsilon_i/(T-1)$. Subtracting $q_i$ in the preceding display and taking reciprocals gives the exact identity
\begin{equation}\label{eq:reciprocal}
     \frac{1}{\tilde\sigma_i^2}-\frac1{q_i}
 =-\frac{2\epsilon_i'x_i\theta_i+\theta_i'x_i'x_i\theta_i}
              {(T-1)q_i\tilde\sigma_i^2}.
\end{equation}

Both $\tilde\sigma_i^2$ and $q_i$ are bounded below by $\epsilon_i'(I_T-P_i)\epsilon_i/(T-1)$. Hence the inverse-moment condition, $\|x_i'\epsilon_i\|=O_p(\sqrt T)$ and $\|x_i'x_i\|=O_p(T)$ give, for an individual unit,
\[
 \frac1{\tilde\sigma_i^2}-\frac1{q_i}
 =\begin{cases}
 O_p(d_{NT}/\sqrt T+d_{NT}^2),& i\in G_1,\\
 O_p((d_{NT}+\gamma^{-1})/\sqrt T+(d_{NT}+\gamma^{-1})^2),&i\in G_2,
 \end{cases}
\]
where we used inverse moment restriction and the result of Lemma 1, i.e.,
\begin{align*}
 \|\theta_i\|=\begin{cases}
 \|\hat\beta_{LS}-\beta\|,& i\in G_1,\\
 \|\lambda/\gamma-(\hat\beta_{LS}-\beta)\|,&i\in G_2 .
 \end{cases}=\begin{cases}
 O_p(d_{NT}),& i\in G_1,\\
 O_p(d_{NT}+\gamma^{-1}),&i\in G_2 .
 \end{cases}
\end{align*}
Furthermore, using $|1/q_i-1/\sigma_i^2|=|q_i-\sigma_i^2|/(q_i\sigma_i^2)$, the positive lower bound on $\sigma_i^2$, and the conditional inverse second moment of $q_i$, Cauchy--Schwarz gives $\E{|1/q_i-1/\sigma_i^2|\mid x_i}\le C/\sqrt T$. Then multiplying the identity \eqref{eq:reciprocal} by $x_i'x_i$ and by $x_i'\epsilon_i$, respectively, and summing separately
over the two groups gives
\begin{align*}
 \frac1{NT}\sum_i\frac{x_i'x_i}{\tilde\sigma_i^2}
 &=\frac1{NT}\sum_i\frac{x_i'x_i}{\sigma_i^2}
   +O_p\!\left(T^{-1/2}+\frac M{N\gamma^2}\right)
   \xrightarrow{p}Q,\\
 \frac1{\sqrt{NT}}\sum_i\frac{x_i'\epsilon_i}{\tilde\sigma_i^2}
 &=\frac1{\sqrt{NT}}\sum_i\frac{x_i'\epsilon_i}{q_i}
   +O_p\!\left(\frac M{\gamma\sqrt{NT}}
               +\frac M{\sqrt N\gamma^2}+\frac1T\right)
   =O_p\!\left(1+\sqrt{N/T}\right).
\end{align*}

Applying the same design bound on $G_2$ gives
$(MT)^{-1}\sum_{i\in G_2}x_i'x_i/\tilde\sigma_i^2\xrightarrow{p}Q^{(G_2)}$.
In the first term of $R_{NT}$, only the $M$ units in $G_2$ enter the two
outer sums. In the second term, the normalized score has order
$O_p(1+\sqrt{N/T})$ by the preceding display. Hence
\[
 \frac{R_{NT}}{\sqrt N}
 =O_p\!\left(\frac{M^2T}{\gamma^2N^{3/2}}
              +\frac1{\sqrt N}+\frac{\sqrt N}{T}\right).
\]
Furthermore, if $M/N\to m_0>0$ and $MT/(\sqrt N\gamma^2)\to c^2$, then
\begin{equation*}
    \frac{1}{\sqrt{N}}R_{NT} = -c^2m_0\lambda'Q^{(G_2)}Q^{-1}Q^{(G_2)}\lambda + o_p(1).
\end{equation*}
By similar arguments
\begin{align*}
    \frac{1}{\sqrt{N}}C_{NT} & =O_p\left(\frac{\sqrt{MT}}{\gamma\sqrt N}
                 +\frac{M\sqrt T}{\gamma N}
                 +\frac{M}{\sqrt N\gamma}
                 +\frac{M\sqrt T}{\sqrt N\gamma^3}\right)=o_p(1),
\end{align*}
and,
\[
    \frac1{\sqrt N}\sum_i\frac{\lambda_i'x_i'x_i\lambda_i}{\tilde\sigma_i^2} =\frac{MT}{\sqrt N\gamma^2}\lambda'Q^{(G_2)}\lambda+o_p(1).
\]
Finally, multiplying the reciprocal identity \eqref{eq:reciprocal} by $\epsilon_i'P_i\epsilon_i$
and summing over the two groups yields
\[
 \frac1{\sqrt N}\sum_i\frac{\epsilon_i'P_i\epsilon_i}{\tilde\sigma_i^2}
 =\frac1{\sqrt N}\sum_i\tilde z_i+
 O_p\!\left(\frac M{\gamma\sqrt{NT}}+\frac M{\sqrt N\gamma^2}+\frac1T\right),
 \qquad
 \tilde z_i=\frac{\epsilon_i'P_i\epsilon_i}{q_i}.
\]
Since $q_i=\epsilon_i'\epsilon_i/(T-1)$, expansion of $1/q_i$ about
$1/\sigma_i^2$ gives $E[\tilde z_i\mid x_i]=K+O(T^{-1})$ uniformly in $i$. Under the additional inverse-moment conditions in \eqref{eq:PYregularity}, the quadratic-form moment arguments in PY (Appendices A.3--A.4) give $N^{-1}\sum_i\operatorname{Var}(\tilde z_i)\to2K$ and $\sup_{i,T}E|\tilde z_i|^{2+\zeta}<\infty$ for some $\zeta>0$. Then putting these rates into expression~\eqref{eq:delta} for the PY test yields
 \begin{align*}
   \Delta=\frac1{\sqrt N}\sum_{i=1}^N\frac{\tilde z_i-K}{\sqrt{2K}}\quad + \frac{MT}{\sqrt N\gamma^2}\,
       \frac{\lambda'\left( Q^{(G_2)}-m_0 Q^{(G_2)}Q^{-1}Q^{(G_2)}\right)\lambda}{\sqrt{2K}} +o_p(1).
 \end{align*}
  where $N^{-1/2}\sum_{i=1}^N(\tilde{z}_{i}-K)/\sqrt{2K}\xrightarrow{d}N(0,1)$ (appendix A.4 of PY) and the error from centering at $K$ is $(2NK)^{-1/2}\sum_{i=1}^N(E[\tilde z_i]-K)=O(\sqrt N/T)=o(1)$. The statement in Theorem~\ref{theorem:Tlarge} follows immediately by considering different scenarios for the rate $MT/(\sqrt N\gamma^2)$.
\end{proof}

\

\begin{proof}[Proof Theorem 1, SC Test]\\
After substituting $\hat{\epsilon}_{i} = \epsilon_i + x_i\lambda_i - x_i(\hat{\beta}_{LS}-\beta)$ into $LM_{SC}-\sqrt{N}\hat{B}$ we obtain,
\begin{equation*}
  LM_{SC}-\sqrt{N}\hat{B} = \sum_{i=1}^N\epsilon_i'(P_i-\diag\{P_i\})\epsilon_i + \sum_{i=1}^N\lambda_i'x_i'(P_i-\diag\{P_i\})x_i\lambda_i + R_{NT} + C_{NT}
\end{equation*}
where $\diag\{A\}$ indicates the square diagonal matrix with main elements given by the diagonal of $A$, and
\begin{align*}
    R_{NT} &= \sum_{i=1}^N(\hat{\beta}_{LS}-\beta)'x_i'(P_i-\diag\{P_i\})x_i(\hat{\beta}_{LS}-\beta) - 2\sum_{i=1}^N\lambda_i'x_i'(P_i-\diag\{P_i\})x_i(\hat{\beta}_{LS}-\beta),\\
    C_{NT} &= 2\sum_{i=1}^N\lambda_i'x_i'(P_i-\diag\{P_i\})\epsilon_i - 2\sum_{i=1}^N\epsilon_i'(P_i-\diag\{P_i\})x_i(\hat{\beta}_{LS}-\beta).
\end{align*}
Using identical arguments to the ones used for BRS test we have,
\begin{align*}
    \frac{1}{\sqrt{N}}R_{NT} &= O_p\left(\frac{M^2T}{{N}^{3/2}\gamma^2}\right)+o_p(1),
\end{align*}
where we explicitly used the fact that $P_i$ is a projection matrix, hence $0\leq P_{i,tt}\leq1$, and
\[
    \sum_{i=1}^N\lVert x_i'\left(P_i-\diag\{P_i\}\right)x_i\rVert_2 \leq \sum_{i=1}^N\sum_{t=1}^T \lVert (x_{it}x_{it}')(1-P_{i,tt})\rVert_2 \leq  \sum_{i=1}^N\sum_{t=1}^T\lVert x_{it}x_{it}'\rVert_2 = O_p(NT).
\]
Here $\lVert A \rVert_2 = \sqrt{\sum_{i=1}^k\sum_{j=1}^k|a_{ij}|^2}$ denotes the Frobenius norm of a $k\times k$ matrix $A$.
When $M/N\to  {m_0} >0 $ and $MT/(\sqrt{N}\gamma^2)\to c^2$ then
\[
    \frac{1}{\sqrt{N}}R_{NT} = c^2m_0\lambda'\left(S^{(G_2)}S^{-1}WS^{-1}S^{(G_2)}- W^{(G_2)}S^{-1}S^{(G_2)} - S^{(G_2)}S^{-1}W^{(G_2)}\right)\lambda + o_p(1).
\]
Similarly,
\begin{equation*}
 \frac{C_{NT}}{\sqrt N}
 =O_p\left(\frac{\sqrt{MT}}{\sqrt N\gamma}
                   +\frac1{\sqrt N}\right) = o_p(1).
\end{equation*}
Further,
\[
  \frac{1}{\sqrt{N}}\sum_{i=1}^N\lambda_i'x_i'(P_i-\diag\{P_i\})x_i\lambda_i \xrightarrow{p} \lambda'W^{(G_2)}\lambda\times \lim_{N,T\to\infty}\frac{MT}{\sqrt{N}\gamma^2}.
\]
Finally, an argument analogous to that used to establish the consistency of the BRS variance estimators yields
\begin{align*}
 &\frac1N\sum_{i=1}^N\sum_{t=2}^T
 \left(\hat\epsilon_{it}\sum_{s<t}P_{i,ts}\hat\epsilon_{is}
       -\epsilon_{it}\sum_{s<t}P_{i,ts}\epsilon_{is}\right)^2=O_p\left[
    T\left\{\frac{M}{N\gamma^2}+d_{NT}^2\right\}
    +\frac{(MT)^{4/3}}{N\gamma^4}
    +N^{1/3}T^{4/3}d_{NT}^4\right]\\
 &\qquad=O_p\left(\frac{a_{NT}}{\sqrt N}+\frac1N
                    +\frac{a_{NT}^2}{(MT)^{2/3}}\right)=o_p(1).
\end{align*}
Thus, by the Cauchy--Schwarz inequality and the law of large numbers,
\begin{align*}
 \hat V_{NT}
 =\frac4N\sum_{i=1}^N\sum_{t=2}^T
       \left(\epsilon_{it}\sum_{s<t}P_{i,ts}\epsilon_{is}\right)^2
       +o_p(1)=\frac4N\sum_{i=1}^N\sigma_i^4\sum_{t=2}^{T}\sum_{s=1}^{t-1}P_{i,ts}^2+o_p(1).
 =\mathcal V_0+o_p(1).
\end{align*}
Moreover, by the central limit theorem,
\begin{equation*}
 \frac1{\sqrt N}\sum_{i=1}^N
       \epsilon_i'(P_i-\operatorname{diag}\{P_i\})\epsilon_i
 \xrightarrow{d}\mathrm{N}(0,\mathcal V_0).
\end{equation*}
Putting all of the obtained rates into expression \eqref{eq:SC} yields the non-centrality parameter and the statement of the theorem.

\end{proof}

\begin{proof}[Remarks on the proof of Corollary~\ref{Cor}]
Under $H_{2,n}$, write $\lambda_i=\gamma^{-1}\sum_{p=2}^{P}\lambda_p \mathbb I\{i\in G_p\}$. The decompositions used in Theorem~\ref{theorem:Tlarge} for three test statistics continues to apply with uniformly bounded $\lambda_p$ (i.e., $\sup_{N,p\le P}\|\lambda_p\|<\infty$) and the stochastic bounds depend on the total number of affected individuals, $\bar{M}$. For example, in the case of PY we have that the direct slope contribution satisfies
\begin{equation*}
 \frac1{\sqrt N\gamma^2} \sum_{p=2}^{P}\lambda_p'B_p\lambda_p 
 =  \frac{\bar{M}T}{\sqrt N\gamma^2}\sum_{p=2}^{P}m_p  \lambda_p'Q^{(G_p)}\lambda_p+o_p(1),
\end{equation*}
where $B^{(G_p)}=\sum_{i\in G_p} x_i'x_i/\tilde\sigma_i^2$, and the result additionally requires the aggregate convergence condition in Assumption \eqref{Ass3}.
Further, for the part of $R_{NT}$ in the proof of Theorem 1 that contributes to the non-centrality parameter we have 
\begin{align*}
 -\frac1{\sqrt N\gamma^2}
 \left(\sum_{p=2}^{P}B_p\lambda_p\right)'
 B^{-1}
 \left(\sum_{q=2}^{P}B_q\lambda_q\right) =
 -\frac{\bar{M}T}{\sqrt N\gamma^2}\bar m
 \sum_{p,q=2}^{P}m_pm_q
 \lambda_p'Q^{(G_p)}Q^{-1}Q^{(G_q)}\lambda_q
 +o_p(1),
\end{align*}
where $B=B^{(\mathcal{N})}$. Combining these two terms and dividing by $\sqrt{2K}$ yields $\delta_\Delta$. The SC and BRS expressions follow from the same expansion of the group-specific deviations around the pooled slope, using their respective moment matrices.
\end{proof}

\section{Formal Statement Fixed $T$}
\begin{assumption}[Regressors Fixed $T$]\label{Ass4}
  $(a)$, $\mathbb{E}|x_{it,k}|^{4+\delta}<C<\infty $, for some $\delta>0$, and all $i=1,...,N$, $t=1,...,T$, and $k=1,...,K$. $(b)$ For any subset $G\subseteq\mathcal{N}$, such that $\#G\to\infty$ as $N\to\infty$, the probability limit $S_T^{(G)}=\plim_{N\to\infty}(\#G)^{-1}\sum_{i\in G}x_{i}'x_{i}$ is a positive definite, non-stochastic and finite matrix.
\end{assumption}
\setcounter{theorem}{1}
\begin{corollary} \label{Cor1}
  Under Assumptions \ref{Ass1}, \ref{Ass4}, and $H_{1,n}$, as $N,M\to\infty$, with $M/N\to {m_0} \in [0,1)$ and $T\geq 2$ fixed, the following holds:
  \begin{equation*}
                LM\overset{d}{\rightarrow }\begin{cases}
                  \chi _{K}^{2},                & \mbox{if } \frac{\sqrt{M}}{\gamma  N^{1/4}} \to 0,      \\
                  \chi _{K}^{2}\left(\delta_{LM}'\mathcal{\psi}^{-1}\delta_{LM} \right), & \mbox{if }  \frac{\sqrt{M}}{\gamma N^{1/4}} \to c>0,  
                \end{cases}
              \end{equation*}
              where $\delta_{LM}$ is a $K\times1$ vector with $k$-th element
              \[
                \delta_{LM,k} = c^2\left(\lambda_2^{\prime}U_k^{(G_2)}\lambda_2 + m_0\lambda_2'\left(S_T^{(G_2)}S_T^{-1}U_k S_T^{-1}S_T^{(G_2)} - 2S_T^{(G_2)}S_T^{-1}U_{k}^{(G_2)}\right)\lambda_2\right),
              \]
              with
              \begin{align*}
                U_{k}^{(G)} = \underset{N \rightarrow \infty}{\text{plim}} \; (2\#G)^{-1} \sum_{i\in G}\sum_{s<t}
                x_{it,k}x_{is,k}\left(x_{it}x_{is}'+x_{is}x_{it}'\right),
              \end{align*}              
              and $U_k=U_k^{(\mathcal{N})}$. Here, $\mathcal{V}$ is a $K \times K$ matrix with $(k,l)$-th element
              \[
              \mathcal{\psi}_{k,l}=\underset{N\to\infty}{\text{plim}}\;N^{-1}\sum_{i=1}^{N}\sum_{t=2}^{T}\sum_{s=1}^{t-1}\sigma_i^4
              x_{it,k}x_{it,l}x_{is,k}x_{is,l}.
              \]
\end{corollary}

\begin{proof}[Proof Corollary 2]
The proof follows as the proof of the BRS test in Theorem \ref{theorem:Tlarge}, with some modifications.
First, under Assumptions \ref{Ass1}, \ref{Ass3}, and $H_{1,n}$, as $N\to \infty$ and $T$ is fixed,
\[
\begin{aligned}    
\hat\beta_{LS} - \beta &= \left(\sum_{i=1}^N\sum_{t=1}^Tx_{it}x_{it}'\right)^{-1}\sum_{i=1}^N\sum_{t=1}^Tx_{it}\epsilon_{it} + \left(\sum_{i=1}^N\sum_{t=1}^Tx_{it}x_{it}'\right)^{-1}\sum_{i=1}^N\sum_{t=1}^Tx_{it}x_{it}^{\prime}\lambda_i\\
&= O_p\left(\frac{M}{N\gamma} + \frac{1}{\sqrt{N}}\right),
\end{aligned}
\]
where we used that $N^{-1}\sum_{i=1}^N x_i'x_i + M^{-1}\sum_{i\in G_2}x_i'x_i = O_p(1)$ and $N^{-1}\sum_{i=1}^N\sum_{t=1}^Tx_{it}\epsilon_{it}=O_p(N^{-1/2})$ by a central limit theorem for the linear regression model. 
First consider sequences satisfying
\[
\frac{M}{N}\longrightarrow m_0\in[0,1),
\qquad
\frac{M}{\sqrt N\gamma^2}\longrightarrow c^2,
\qquad c\in[0,\infty).
\]
Let $R_{N,k}$ and $C_{N,k}$ denote the fixed-$T$ counterparts of
$R_{NT,k}$ and $C_{NT,k}$ in Theorem~\ref{theorem:Tlarge}.
Using the same score decomposition, the above behaviour of $\hat\beta_{LS}$, and conditional second-moment bounds under Assumptions~\ref{Ass1} and~\ref{Ass3}, we obtain
\begin{align*}
    \frac{1}{\sqrt{N}}R_{N,k} & = c^2m_0\lambda'\left(S_T^{(G_2)}S_T^{-1}U_kS_T^{-1}S_T^{(G_2)} - 2S_T^{(G_2)}S_T^{-1}U_{k}^{(G_2)}\right)\lambda + o_p(1),
\end{align*}
and $\frac{1}{\sqrt{N}}C_{N,k} = o_p\left(1\right)$. Furthermore,
  \begin{align*}
  \frac{1}{\sqrt{N}}\sum\limits_{i=1}^{N}\sum\limits_{t=2}^{T}\sum\limits_{s=1}^{t-1}\lambda_i'{x}_{it}x_{it,k}x_{is,k}{x}_{is}'\lambda_i &=\frac{1}{\gamma^2\sqrt{N}} \sum_{i\in G_2}\sum\limits_{t=2}^{T}\sum\limits_{s=1}^{t-1}\lambda'{x}_{it}x_{it,k}x_{is,k}{x}_{is}'\lambda\\ &\overset{p}{\rightarrow}\lambda' U_k^{(G_2)}\lambda\times \lim\limits_{N\to\infty} \frac{M}{\sqrt{N}\gamma^2}.
  \end{align*}
  and, by the argument in the proof Corollary 5.2 in BRS,
  \begin{eqnarray*}
    \frac{1}{\sqrt{N}}s \overset{d}{\rightarrow}N(0,\mathcal{V})  \text{ as } N\to\infty.
  \end{eqnarray*}
  Finally, following the variance-estimator argument associated with Remark 6.2 in BRS, adapted to the present sequence of local alternatives, we can write each $(k,l)$ element of $\hat{V}$ as,
  \begin{eqnarray*}
      \frac{1}{N}\hat{V}_{k,l} &= \frac{1}{N} \sum\limits_{i=1}^N\sum\limits_{t=2}^T\epsilon_{it}^2x_{it,k}x_{it,l}\left(\sum\limits_{s=1}^{t-1}\epsilon_{is}x_{is,k}\right)\left(\sum\limits_{p=1}^{t-1}\epsilon_{ip}x_{ip,l}\right) + o_p\left(1 \right)\\
      &= \mathcal{\psi}_{k,l} + o_p\left(1\right).
  \end{eqnarray*}
 Putting all these results together yields the corollary.
\end{proof}

\section*{Declaration of competing interest}
The authors declare that they have no known competing financial interests or personal relationships that could have appeared to influence the work reported in this paper.

\section*{Acknowledgments}
The authors thank Robin Lumsdaine and two anonymous referees for helpful discussions and comments.

\section*{Funding}
Nazarii Salish gratefully acknowledges financial support from Juan de la Cierva Incorporaci\'{o}n [IJC2019-041742-I], Ramon y Cajal [RYC2023-044391-I], and CEX 2021-001181-M, financed by MICIU/AEI /10.13039/501100011033.

\section*{Data availability}
No empirical data were used in this study; all numerical results are based on simulated data.

\bibliographystyle{abbrvnat}
\bibliography{EHP}

\end{document}